\documentclass{article}

\usepackage[ansinew]{inputenc}
\usepackage[T1]{fontenc}

\usepackage{anysize}
   
\usepackage{graphicx}
\usepackage{amssymb}
\usepackage{textcomp}
\usepackage{amscd}
\usepackage{amsthm}
\usepackage{amsmath}
\usepackage[usenames]{color}
\definecolor{sonnengelb}{rgb}{0.9,0.5,0} 
\definecolor{brown}{rgb}{0.5,1,0}
\usepackage{algorithm,algorithmic}

\title{Uniform Sampling of Undirected and Directed Graphs with a Fixed Degree
  Sequence\thanks{This work was partially supported by the DFG Focus
    Program Algorithm Engineering, grant Mu 1482/4-1, and by 
a VolkswagenStiftung grant for the project ``Impact on motif content on
dynamic function of complex networks''.}}
\author{Annabell Berger and Matthias M\"uller-Hannemann\\[1ex]
  Department of Computer Science\\
  Martin-Luther-Universit\"at Halle-Wittenberg\\
  \{berger,muellerh\}@informatik.uni-halle.de}
\date{}

\newtheorem{Definition}{Definition}[section]
\newtheorem{Example}{Example}[section]
\newtheorem{Theorem}{Theorem}[section]

\newtheorem{Lemma}[Theorem]{Lemma}
\newtheorem{Proposition}[Theorem]{Proposition}
\newtheorem{Corollary}[Theorem]{Corollary}

\begin{document}

\thispagestyle{empty}
\maketitle

\begin{abstract}
Many applications in network analysis require algorithms 
to sample uniformly at random from the set of all graphs  
with a prescribed degree sequence.
We present a Markov chain based approach which 
converges to the uniform distribution of all realizations
for both the directed and undirected case. 
It remains an open challenge whether these Markov chains are rapidly mixing.

For the case of directed graphs, we also explain in this paper 
that a popular switching algorithm fails in general to
sample uniformly at random because the state graph of the 
Markov chain decomposes into
different isomorphic components. 
We call degree sequences for which the state
graph is strongly connected  \emph{arc swap sequences}. 
To handle arbitrary degree sequences, we develop two different
solutions. The first uses an additional operation (a reorientation of
induced directed 3-cycles) which makes the state graph strongly
connected, the second selects randomly one of the isomorphic
components and samples inside it.
Our main contribution is a precise characterization of arc swap sequences, 
leading to an efficient recognition algorithm.
Finally, we point out some interesting consequences for network analysis.
\end{abstract}

\section{Introduction}

We consider the problem of sampling uniformly at random from the set
of all realizations of a prescribed degree sequence as simple, 
labeled graphs or digraphs, respectively, without loops.

\paragraph{Motivation.}
In complex network analysis, one is interested in studying certain network
properties of some observed real graph in comparison with an ensemble
of graphs with the same degree sequence to detect deviations from 
randomness~\cite{Milo-etal04}. For example, this is used to study the
motif content of classes of networks~\cite{Milo-et-al04b}.
To perform such an analysis, a uniform sampling from the set of all
realizations is required.
A general method to sample random elements from some set of objects is
via rapidly mixing Markov chains~\cite{Sinclair92,Sinclair93}. 
Every Markov chain can be viewed as a random walk on a directed graph,
the so-called \emph{state graph}.
In our context, its vertices (the states) correspond one-to-one to the
set of all realizations of prescribed degree sequences. 
For a survey on random walks, we refer to Lov{\'a}sz~\cite{Lovasz96}.

A popular variant of the Markov chain approach 
to sample among such realizations is the so-called 
\emph{switching-algorithm}. It starts with a given realization, and
then performs a sequence of 2-swaps.

In the undirected case, a \emph{2-swap} replaces two non-adjacent edges
$\{a,b\}, \{c,d\}$ either by $\{a,c\}, \{b,d\}$ or by $\{a,d\},
\{b,c\}$, provided
that both new edges have not been contained in the graph before the
swap operation. Likewise, in the directed case, 
given two arcs $(a,b), (c,d)$ with all vertices $a,b,c,d$ being
distinct, a \emph{2-swap} replaces these two arcs by $(a,d), (c,b)$
which are currently not included in the realization (the latter is
crucial to avoid parallel arcs).
The switching algorithm is usually stopped heuristically 
after a certain number of iterations, and then outputs the resulting
realization as a ``random element''.  

For undirected graphs, one can prove that this switching algorithm
converges to a random stage. The directed case, however, turns out to
be much more difficult. The following example demonstrates that
the switching algorithm does not even converge to a random stage.

\begin{Example}\label{example:disconnectedness}
Consider the following class of digraphs $D=(V,A)$ with $3n$ vertices
$V = \{ v_1, v_2, \dots, v_{3n}\}$, see
Figure~\ref{fig:disconnectedness}. 
Roughly speaking, this class
consists of induced directed 3-cycles $C_i$ formed by triples 
$V_i = \{ v_{3i}, v_{3i+1}, v_{3i+2}\}$  of vertices,
and arcs $A_i = \{ (v_{3i}, v_{3i+1}), (v_{3i+1},v_{3i+2}), (v_{3i+2},v_{3i}) \}$
for $i \in \{0, \dots, n-1\}$.
All vertices of cycle $C_i$ are connected to all other vertices 
of cycles with larger index than $i$.
More formally, let $A' := \{ (v,w) | v \in V_i, w \in V_j, i < j\}$. 
We set $A := A' \cup (\cup_{i=1}^n \, A_i)$.

It is easy to check that no 2-swap can be applied to this digraph.
However, we can independently reorient each of the $n$ induced
$3$-cycles, leading to $2^{n/3}$ many (isomorphic) realizations of the 
same degree sequence. Thus, if we use a random walk on the state graph of
all realizations of this degree sequence and 
use only 2-swaps to define the possible
transitions between realizations, this state graph 
consists exactly of $2^{n/3}$ many singleton components.  
Hence, a ``random walk'' on this graph will be stuck in a single
realization although exponentially many realizations exist.
\end{Example}

\begin{figure}[t]
\centerline{\includegraphics[height=2.5cm]{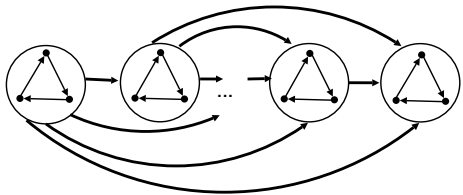}}
\caption{\label{fig:disconnectedness} Example of digraphs where no
  2-swap operation can be applied.}
\end{figure}

More examples of graph classes of this type will be given in the Appendix.
It is interesting to note that 2-swap operations suffice to sample 
directed graphs \emph{with} loops as has been proven by 
Ryser~\cite{Ryser57} in the context of square matrices 
with $\{0,1\}$--entries which
can be interpreted as node-node adjacency matrices of digraphs with loops.

\paragraph{Realizability of degree sequences.}
In order to use a Markov chain approach one needs at least one
feasible realization. In applications from complex network analysis, 
one can usually take the degree sequence of some observed real world
graph. Otherwise, one has to construct a realization.

The realization problem, i.e., characterizing the existence and
finding at least one realization, has quite a long history.
First results go back to the seminal work by Tutte who solved the more
general $f$-factor problem~\cite{Tutte52}. 
Given a simple graph $G=(V,E)$ and a function 
$f : V (G) \mapsto \mathbb{N}_0$, an \emph{$f$-factor} is a subgraph
$H$ of $G$ such
that every vertex $v\in V$ in this subgraph $H$ has exactly degree 
$d_G(v)=f(v)$. 
Tutte gave a polynomial time transformation of the $f$-factor problem
to the perfect matching problem. This implies the first polynomial time
algorithm for finding some $f$-factor~\cite{Tutte54}.
For a survey on efficient algorithms for the $f$-factor problem
by matching or network flow techniques, we refer to Chapter 21 of 
Schrijver~\cite{Schrijver03}. 
Clearly, if the given graph $G$ is complete, then every $f$-factor is a
solution of the degree sequence problem.
Erd\H{o}s and Gallai~\cite{ErdosGallai60} proved a simpler Tutte-type
result for the degree sequence problem. 
Already in 1955,  Havel~\cite{Havel55} 
developed a simple greedy-like algorithm to construct a realization of a given
degree sequence as a simple undirected graph without loops. A few years later,
Hakimi~\cite{Hakimi62,Hakimi65} studied the simpler case of undirected
graphs with multiple edges. 

It is also well-known how to test
whether a prescribed degree sequence can be realized as a digraph.
Chen~\cite{Chen66} presented necessary and sufficient conditions for
the realizability of degree sequences which can be checked in linear time. 
Again, the construction of a concrete realization is equivalent to
an $f$-factor problem on a corresponding undirected bipartite graph.
Kleitman and Wang~\cite{KleitmanWang73}
found a greedy-type algorithm
generalizing previous work by Havel~\cite{Havel55} and 
Hakimi~\cite{Hakimi62,Hakimi65}. 
This approach has recently been rediscovered by
Erd\H{o}s et al.~\cite{Erdos-etal09}.

\paragraph{Related work.}

Kannan et al.~\cite{KannanTetaliVempala99} showed how to sample
bipartite undirected graphs via Markov chains. 
They proved polynomial
mixing time for regular and near-regular graphs.
Cooper et al.~\cite{CooperDyerGreenhill07} extended this work to
non-bipartite undirected,
$d$-regular graphs and proved a polynomial mixing time for the 
switching algorithm. More precisely, they upper bounded
the mixing time in these cases by $d^{15} n^8 (dn \log(dn) +
\log(\varepsilon^{-1}))$, for graphs with $|V| = n$. 
In a break-through paper,
Jerrum, Sinclair, and Vigoda~\cite{JerrumSinclairVigoda04} presented
a polynomial-time almost uniform sampling algorithm for perfect
matchings in bipartite graphs. Their approach can be used to sample
arbitrary bipartite graphs and arbitrary 
digraphs with a specified degree sequence in $O( 
n^{14} \log^4 n)$ via the above-mentioned reduction due to Tutte.
In the context of sampling binary contingency tables, 
Bez{\'a}kov{\'a} et al.~\cite{BezakovaBhatnagarVigoda07} managed to improve
the running time for these sampling problems to 
$O(n^{11} \log^5 n)$, which is still far from practical.


McKay and Wormald~\cite{McKayWormald90,McKayWormald91} 
use a configuration model and generate 
random undirected graphs with degrees bounded by
$o(n^{1/2} )$
with uniform distribution in $O(m^2 d_{max} )$ time,
where $d_{max}$ denotes the maximum degree, and $m$ the number of edges.
Steger and Wormald~\cite{StegerWormald99} introduced 
a modification of the configuration model that leads to a fast algorithm
and samples asymptotically uniform for degrees up to $o(n^{1/28} )$. 
Kim and Vu ~\cite{KimVu03} improved the analysis of Steger and Wormald's
algorithm, proving that the
output is asymptotically uniform for degrees up to
$O(n^{1/3-\varepsilon})$, for any $\varepsilon > 0$.
Bayati et al.~\cite{Bayati-etal09} recently presented a
nearly-linear time algorithm for counting and randomly generating 
almost uniformly
simple undirected graphs with a given degree sequence where
the maximum degree is restricted to $d_{max} = O(m^{1/4-\tau})$,
and $\tau$ is any positive constant.

\paragraph{Random walks and Markov chains.}
Let us briefly review the basic notions of random walks and
their relation to Markov
chains. See \cite{Lovasz96,JerrumSinclair97,Sinclair93} for more details.
A random walk (Markov chain) on a digraph $D=(V,A)$
 is a sequence of vertices
 $v_0,v_1,\dots,v_t,\dots$ where $(v_i,v_{i+1})\in A$. Vertex $v_0$
 represents the initial state. Denote by $d_D^+(v)$ the out-degree of vertex
 $v \in V$.
 At the $t$th step we move to an
 arbitrary neighbor of $v_{t}$ with probability $1/{d_D^+(v_t)}$ or
 stay at $v_{t}$ with probability $(1-\nu(v_t))/{d_D^+(v_t)},$ where $\nu(v_t)$
 denotes the number of neighbors of $v_t.$ Furthermore, we define the
 distribution of $V$ at time $t\in \mathbb{Z}^+$ as the function
 $P_{t} \in [0,1]^{|V|}$ with $P_{t}(i):= Prob(v_t=i).$ 
A well-known
 result \cite{Lovasz96} is that $P_t$ tends to the uniform stationary
 distribution for $t\rightarrow \infty$, if the digraph is 
 (1) non-bipartite (that means aperiodic), (2) strongly connected (i.e., irreducible), 
 (3) symmetric, and (4) regular.
 A digraph $D$ is \emph{$d_D$-regular} if all vertices have the same in- and
out-degrees~$d_D$.

In this paper, we will view all Markov chains as random walks on
symmetric $d_D$-regular digraphs $D=(V,A)$ 
whose vertices correspond to the state space $V$. 
The transition probability on each arc $(v,w) \in A$ will be the
constant $1/d_D$.

\paragraph{Our contribution.}
In this paper, we prove the following results.
\begin{itemize}
\item
For undirected graphs we analyze the well-known switching algorithm.
It is straight-forward to translate the switching algorithm  
into a random walk on an appropriately defined Markov chain.
This Markov chain corresponds to a symmetric, regular, strongly
connected, non-bipartite simple digraph with directed loops allowed. 
Thus, it converges to the uniform distribution of all realizations.
Each realization of the degree sequence
is a vertex of this digraph, and two realizations are mutually connected  
by arcs if and only if their symmetric difference is an alternating
4-cycle (i.e., corresponds to a 2-swap).
This graph becomes regular by adding additional loops, see
Section~\ref{sec:undirected}. 

Cooper et al.~\cite{CooperDyerGreenhill07} already showed in the
context of regular graphs that the underlying digraph of 
this Markov chain is strongly
connected, but we give a much simpler proof of this property.
Its diameter is bounded by the number $m$ of edges
in the prescribed degree sequence.  

\item
Carefully looking at our Example~\ref{example:disconnectedness},
we observe that in the directed case the state graph becomes 
strongly connected if we add a second type of operation to transform 
one realization into another: Simply reorient the arcs of an induced
directed 3-cycle. We call this operation \emph{3-cycle reorientation}.
We give a graph-theoretical proof
that 2-swaps and 3-cycle
reorientations suffice not only here, but also in general for
arbitrary prescribed degree sequences.
These observations allow us to define a 
Markov chain, very similar to the undirected case. 
The difference is that
two realizations are mutually connected  
by arcs if and only if their symmetric difference is either an
alternating directed 4-cycle or 6-cycle with exactly three different vertices.
Again, this digraph becomes regular by adding additional loops, see
Section~\ref{sec:general-markov}. The transition probabilities are of
order $O(1/m^2)$, and
the diameter can be bounded by $O(m)$, where $m$ denotes 
the number of arcs in the prescribed degree sequence. 

In the context of $(0,1)$-matrices with given marginals (i.e.,
prescribed degree sequences in our terminology), 
Rao et al.~\cite{RaoJanaBandyopadhyay96} similarly observed that
switching operations on so-called ``compact alternating hexagons'' are
necessary. A compact alternating hexagon is a  
$3\times 3$-submatrix, which can be interpreted as the adjacency matrix
of a directed $3$-cycle subgraph. They define a random walk on a
series of digraphs, starting with a non-regular state graph
which is iteratively updated towards regularity, i.e. their Markov
chain converges asymptotically to the uniform distribution. 
However, it is unclear how fast this process converges and whether
this is more efficient than starting directly with a single regular
state graph. Since Rao et al.\ work directly on
matrices, their transition probabilities are of order $O(1/n^6)$,
i.e., by several orders smaller than in our version.

Very recently, Erd\H{o}s et al.~\cite{Erdos-etal09} proposed a similar
Markov chain approach using 2-swaps and 3-swaps. The latter type of
operation exchanges a simple directed 3-path or 3-cycle $(v_1,v_2),
(v_2,v_3)$, $(v_3,v_4)$ (the first and last vertex may be identical) by
$(v_1,v_3), (v_3, v_2), (v_2,v_4)$, but is a much larger set of operations
than ours. 

\item
Although in directed graphs 2-swaps alone do not suffice to sample uniformly in general, the
corresponding approach is still frequently used in network analysis.
One reason for the popularity of this approach --- 
in addition to its simplicity --- 
might be that it empirically worked in many cases quite 
well~\cite{Milo-etal04}.
In this paper, we study under which conditions this approach 
can be applied and provably leads to correct uniform sampling.
We call such degree sequences \emph{arc-swap sequences}, and 
give a graph-theoretical characterization which can be checked in
polynomial time. More specifically, we can recognize arc-swap
sequences in $O(m^2)$ time using matching techniques.
Using a parallel Havel-Hakimi algorithm by LaMar~\cite{LaMar09},
originally developed to realize Euler sequences with an odd number of
arcs, the recognition problem
can even be solved in linear time. 
This algorithm also allows us to determine the  
number of induced directed $3$-cycles 
which appear in \emph{every} realization. 

However, the simpler approach comes with a price: 
our bound on the diameter of the state graph becomes $m n$ and so 
is by one order of $n$ worse 
in comparison with using 2-swaps and 3-cycle reorientations.
Since half of the diameter is a trivial lower bound 
on the mixing time
and the diameter also appears as a factor in known upper bounds, 
we conjecture that the classical switching
algorithm requires a mixing time $\tau_{\varepsilon}$ with
an order of $n$ more steps as the variant with
$3$-cycle reorientation.

In those cases where 2-swaps do not suffice
to sample uniformly, the state graph decomposes into $2^k$
strongly connected components, where $k$ is
the  number of induced directed $3$-cycles 
which appear in every realization. 
We can also efficiently determine the number of strongly
connected components of the state graph (of course, without
explicitly constructing this exponentially sized graph). 
However, all these components are isomorphic.
This can be exploited as follows: 
For a non-arc-swap sequence, 
we first determine all those induced directed $3$-cycles 
which appear in every realization. By reducing the in- and
out-degrees for all vertices of these $3$-cycles by one, we then obtain
a new sequence, now guaranteed to be an arc-swap sequence. On the
latter we can either use the switching algorithm or 
our variant with additional $3$-cycle reorientations on a smaller
state graph with a reduced diameter $n (m-3k)$ or $m-3k$, 
respectively, yielding an important practical advantage.

Our results 
give a theoretical foundation to compute certain network
characteristics on unlabeled digraphs in a single component using
2-swaps only. For example, this includes the analysis of the 
motif content~\cite{Milo-et-al02}. Likewise
we can still compute the average diameter among all realizations if 
we work in a single component. However, for other network
characteristics, for example betweenness centrality 
on edges~\cite{Koschuetzki-et-al05}, 
this leads in general to incorrect estimations.

\end{itemize}

\paragraph{Overview.} The remainder of the paper is structured as follows.
In Section~\ref{sec:undirected}, we start with the undirected case. 
We introduce  appropriately defined state graphs  
underlying our Markov chains, and show for these
graphs crucial properties like regularity and strong connectivity.  
We also upper bound their diameter.
The more difficult directed case is presented in 
Section~\ref{sec:general-markov}.
Afterwards, in Section~\ref{sec:arc-swap-sequences}, we characterize
those degree sequences for which a simpler Markov chain based on
2-swaps provably leads to uniform sampling in the directed case. 
We also describe a few consequences and applications. 
Finally, we conclude with a short summary and remarks on future work.

\section{Sampling Undirected Graphs}\label{sec:undirected}

In this section we show how to sample undirected graphs with
a prescribed degree sequence uniformly at
random with a  random walk. 
This section is structured as follows.
We first give a formal problem definition and introduce some notation. Then 
we introduce an appropriately defined Markov chain and prove that it
has all desired properties.

\paragraph{Formal problem definition.}
In the undirected case, a degree sequence $S$ of order $n$ is the ordered set
$(a_1, a_2, \dots, a_n)$ with $a_i \in \mathbb{Z}^+, a_i > 0$. 
Let $G=(V,E)$ be an
 undirected labeled graph $G=(V,E)$ without loops and parallel edges and
 $|V|=n$. We define the \emph{degree-function} $d:V\rightarrow
 \mathbb{Z}^+$ which assigns to each vertex $v_i\in V$ the number of
 incident edges. We call $S$ a
 \emph{graphical sequence} if and only if there exists at least one
 undirected labeled graph $G=(V,E)$ without any loops or parallel edges
 which satisfies
 $d(v_i)=a_i$ for all $v_{i}\in V$ and $i\in \{1,\dots,|V|\}.$ 
 Any such undirected graph $G$ is called
 \emph{realization} of $S$.

We define an \emph{alternating  walk} $P$ for a
graph $G=(V,E)$ as a sequence $P:=(v_1,v_2,\dots, v_{\ell})$ of
vertices $v_i\in V$ where either $\{v_i,v_{i+1}\}\in E(G)$ and
$\{v_{i-1},v_{i}\}\notin E(G)$  or $\{v_i,v_{i+1}\}\notin E(G)$ and
$\{v_{i-1},v_{i}\}\in E(G)$ for $i\mod{2}=1.$ 
The length of a walk (or path, cycle, respectively) is the number of
its edges. We call an  
alternating
walk $C$ of even length \emph{alternating cycle} if $v_1=v_{\ell}$ is
fulfilled. 
For two realizations $G, G'$, the symmetric difference of their edge
sets $E(G)$ and $E(G')$ is denoted as
$G\Delta G' := (E(G) \setminus E(G')) \cup (E(G') \setminus E(G))$.
A graph is called \emph{Eulerian} 
if every vertex has even degree.
Note that the symmetric difference $G\Delta G'$ of two realizations
$G,G'$ is Eulerian and hence 
always decomposes into a number of alternating cycles.

\paragraph{The Markov chain.}

We denote by $\Psi = (V_{\psi}, A_{\psi})$ the digraph for our random
walk, the \emph{state graph}, for short. 
Its underlying vertex set $V_{\psi}$ is the set of
all realizations of a given degree sequence $S$.
For a realization $G$, we denote by $V_G$ the
corresponding vertex in $V_{\psi}$.
The arc set $A_{\psi}$ is defined as follows.
\begin{enumerate}
\item[a)] We connect two vertices $V_G,V_{G'} \in V_{\psi}, G\neq G'$ with arcs $(V_G,V_{G'})$ and $(V_{G'},V_G)$ if and only if $|G\Delta G'|=4$ is fulfilled.
\item[b)] We set for each pair of non-adjacent edges $\{v_{i_{1}},v_{i_{2}}\},\{v_{i_{3}},v_{i_{4}}\}\in E(G), i_{j}\in \{1,\dots,n\}$ a directed loop $(V_G,V_G)$ if and only if 
 $\{v_{i_{1}},v_{i_{4}}\}\in E(G) \lor \{v_{i_{3}},v_{i_{2}}\}\in
 E(G)$.
\item[c)] We set for each pair of non-adjacent edges $\{v_{i_{1}},v_{i_{2}}\},\{v_{i_{3}},v_{i_{4}}\}\in E(G), i_{j}\in \{1,\dots,n\}$ a directed loop $(V_G,V_G)$ if and only if 
 $\{v_{i_{1}},v_{i_{3}}\}\in E(G) \lor \{v_{i_{2}},v_{i_{4}}\}\in
 E(G)$.
\item[d)] We set one directed loop $(V_G,V_G)$ for all $V_G\in V_{\psi}.$ 
\end{enumerate}

\begin{Lemma}\label{lemma:state-graph-undirected-basic-properties}
 The state graph $\Psi = (V_{\psi}, A_{\psi})$ is non-bipartite, 
symmetric, and regular. 
\end{Lemma}

\begin{proof}
Non-bipartiteness follows from the insertion of directed loops.
Likewise, symmetry is obvious since we always introduce arcs in both
directions in case a). 
For each pair of non-adjacent edges of a realization  $G$, 
we introduce exactly two arcs in $\Psi$. These arcs either connect two
neighboring states or are directed loops.
Thus each vertex $V_G \in V_{\psi}$ has an out-degree of twice the number
of non-adjacent edges in $G$ plus one (for the loop in step d)). 
Due to symmetry, the out-degree equals
the in-degree. For each realization $G$, the number of pairs of
non-adjacent edges is exactly ${|E| \choose 2} - \sum_{v_i \in V(G)}
{d_G(v_i) \choose 2} = {|E| \choose 2} - \sum_{i=1}^n \; {a_i \choose 2}$, that
is a constant independent of $G$.
\end{proof}

The next step is to show that the state graph is strongly connected.
We first prove the following auxiliary proposition which asserts that
the symmetric difference of two different realizations always
contains a vertex-disjoint path of length three.

\begin{Proposition}\label{prop:3-path}
Let $S$ be a graphical sequence and $G$ and $G'$ be two different
realizations, i.e., $G\Delta G' \neq \emptyset$.
Then there exists  a vertex-disjoint alternating walk
$P=(v_1,v_2,v_3,v_4)$ in $G\Delta G'$ 
with $\{v_1,v_2\},\{v_3,v_4\}\in E(G)$ and $\{v_2,v_3\}\in E(G').$
\end{Proposition}

\begin{proof}
In the proof of this proposition, we argue only about edges in the
symmetric difference $G\Delta G'$ which is assumed to be non-empty. 
Therefore, there are edges $\{v_1, v_2\}, \{v_2,
v_3\}$ with $\{v_1,v_2\} \in E(G)$ and $\{v_2,v_3\} \in
E(G')$ and $v_1 \neq v_3$. If there is also an edge $\{v_3,v_4\} \in
E(G)$ with
$v_4\neq v_1$, we are done with the vertex-disjoint 
alternating walk $P = (v_1,v_2,v_3,v_4)$ as desired.
Otherwise, the symmetric difference must contain the edge $\{v_3,v_1\}
\in E(G)$ and also some edge $\{v_1,v_4\} \in E(G')$. Note that $v_4
\neq v_2$ and $v_4 \neq v_3$. This implies the existence of another
edge $\{v_4,v_5\} \in E(G)$. Note also that $v_5 \neq v_3$, since we
are in the case that $\{v_3,v_4\}$ does not exist. 
Either $v_5 = v_2$ or $v_5$ is a new vertex disjoint from $\{ v_1,
\dots, v_4\}$. 
Therefore, in both cases $P =
(v_3,v_1,v_4,v_5)$ is a vertex-disjoint alternating walk composed of
edges from the symmetric difference.     
\end{proof}

\begin{Lemma}\label{lemma:undirected-connectedness} 
Let $S$ be a graphical sequence and let $G \neq G'$ be two realizations. 
Then there exist 
realizations $G_0,G_1,\dots,G_k$ with $G_0:=G$, $G_k:=G'$ and 
 $|G_i\Delta G_{i+1}|=4$ 
where $k\leq \frac{1}{2}|G\Delta G'|-1.$ 
\end{Lemma}

\begin{proof} 
We prove the lemma by induction according to the cardinality of the
symmetric difference $|G\Delta G'|=2\kappa.$ 
For $\kappa:=2$ we get $|G\Delta G'|=4.$ The correctness of our claim follows with $G_1:=G'$. 
We assume the correctness of our claim for all $\kappa\leq \ell.$ 
Consider $|G\Delta G'|=2\ell +2$. 
According to Proposition~\ref{prop:3-path}, there exists in $G\Delta
G'$  an alternating vertex-disjoint walk $P=(v_1,v_2,v_3,v_4)$ with
$\{v_1,v_2\},\{v_3,v_4\}\in E(G)$ and $\{v_3,v_2\}\in E(G').$ 
\begin{enumerate}
 \item[case 1:] Assume $\{v_1,v_4\}\in E(G')\setminus E(G).$\\
This implies $\{v_1,v_4\}\in G\Delta G'$.
$G_1:=(G_0\setminus \{\{v_1,v_2\},\{v_3,v_4\}\})\cup \{\{v_2,v_3\},\{v_1,v_4\}\}$ is a realization of $S$ and it follows
$|G_0\Delta G_1|=4$ and $|G_1\Delta G'|=2\ell+2-4=2(\ell-1).$ 
Note that after this step, $G_1\Delta G'$ may consist of several
connected components, but each of them has strictly smaller
cardinality. 
Thus, we obtain realizations $G_1,G_2,\dots,G_k$ with $G_k:=G'$ and $|G_i \Delta G_{i+1}|=4$  where $k-1\leq \frac{1}{2}|G_1\Delta G'|-1.$ Hence, we get the sequence $G_0,G_1,\dots,G_k$ with $k= 1+\frac{1}{2}|G_1\Delta G'|-1=\frac{1}{2}(|G\Delta G'|-4)\leq \frac{1}{2}|G\Delta G'|-1$.

\item[case 2:] Assume $\{v_1,v_4\}\in E(G)\cap E(G').$\\
This implies $\{v_1,v_4\}\notin G\Delta G'.$
$P$ is an alternating subpath of an alternating cycle
$C=(v_4,v_i,\dots,v_j,$ $v_1, v_2, v_3, v_4)$ with $\{v_i,v_4\},\{v_1,v_j\}\in A(G')$ of
length $|C|\geq 6$. We construct a new alternating cycle
$C^*:=(C\setminus P)\cup \{\{v_1,v_4\}\}$ with length
$|C^*|=|C| - 2\geq 4$. 
We swap the arcs in $C^*$ and get a realization $G^*$ of S with
$|G_0\Delta G^*|=|C^*|\leq 2\ell$ and $|G^*\Delta G'|=|G \Delta
G'|-(|{C^*}|-1)+1\leq 2\ell.$ The symmetric difference 
$G^*\Delta G'$ may consist of several
connected components, but their total length is bounded by $2\ell$.
Thus there exist 
sequences 
$G_0^{1}:=G,G_1^{1},\dots,G_{k_1}^1:=G^*$ and
$G_0^{2}:=G^*,G_1^{2},\dots,G_{k_2}^2=G'$ with $k_1\leq
\frac{1}{2}|G_0\Delta G^*|-1 = \frac{1}{2} |C^*|-1$ 
and $k_2 \leq \frac{1}{2} |G^* \Delta G'|-1 = \frac{1}{2} (|G\Delta G'| -
(|C^*|-1)+1)-1 \leq 2 \ell$. 
We arrange these sequences
one after another and get a sequence which fulfills $k=k_1+k_2 
= \frac{1}{2} |C^*|-1 + \frac{1}{2} (|G\Delta G'| -
(|C^*|-1)+1)-1  = \frac{1}{2}(|G\Delta G'|) - 1$.

\item[case 3:] Assume $\{v_1,v_4\}\in E(G)\setminus E(G').$\\
This implies $\{v_1,v_4\}\in G\Delta G'.$
Assume first that the symmetric difference $G\Delta G'$ 
contains an alternating cycle $C$ which avoids $P$. 
Then, we can apply the induction hypothesis to $C$. 
Swapping the edges of $C$, we get a realization $G^*$ of sequence $S$
 with
$|G_0\Delta G^*|=|{C^*}|\leq 2\ell$ and $|G^*\Delta G'|=|G \Delta
G'|-|C|\leq 2 \ell$. 
According to the induction hypothesis there exist sequences
$G_0^{1}:=G,G_1^{1},\dots,G_{k_1}^1:=G^*$ and
$G_0^{2}:=G^*,G_1^{2},\dots,G_{k_2}^2=G'$ with $k_1\leq
\frac{1}{2}|G_0\Delta G^*|-1$ and $k_2 \leq \frac{1}{2} |G^* \Delta G'|-1$. 
We arrange these sequences
one after another and get a sequence which fulfills
$k=k_1+k_2 \leq \frac{1}{2}|G_0\Delta G^*|-1+ \frac{1}{2} |G^* \Delta G'|-1
\leq\frac{1}{2}(|G\Delta G'|)-1$.

It remains to consider the case that such a cycle $C$ does not exist.
In other words, every alternating cycle in $G\Delta G'$ 
includes edges from $P$.

The alternating walk $P$ can be extended to an alternating cycle $C^* =
( v_1, v_2, v_3, v_4, v_5, \dots, v_{2t}, v_1)$, $t\geq 3$ using
only arcs from $G\Delta G'$. To construct $C^*$, start with $P$, and keep
adding alternating edges until you reach the start vertex $v_1$ for the
first time with an edge $\{v_i, v_1\} \in E(G')$.
Since the symmetric difference is Eulerian, you will not
get stuck before reaching $v_1$ with such an edge.
Note that $C^*$ must contain the edge $\{v_1,v_4\}$, as otherwise an
alternating cycle of type $C$ would exist. 
This also implies the existence of an alternating sub-walk $P_1 = (
v_4, v_5, \dots, v_6, v_4)$ of $C^*$ of odd length 
(at least of length 3), 
starting and ending with edges in $E(C^*)$. Likewise, there must be
another alternating sub-walk $P_2 = \{ v_1, v_7, \dots, v_8, v_1 \}$,
also of odd length (at least of length 3), starting and ending 
with edges in $E(C^*)$. The situation is visualized in
Figure~\ref{fig:proof-detail-case3}.

\begin{figure}[t]
\centerline{\includegraphics[height=2.5cm]{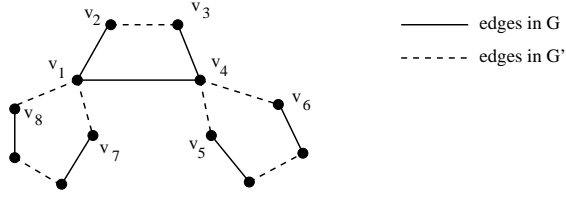}}
\caption{\label{fig:proof-detail-case3} Proof of 
Lemma~\ref{lemma:undirected-connectedness}: 
Edges of the symmetric difference $G \Delta G'$ in case 3.}
\end{figure}

In this scenario, we have $v_5 \neq v_7$, as otherwise $(E(P_1)
\setminus \{ \{v_7,v_1\} \}) \cup \{\{v_7=v_5, v_4\}, \{v_4,v_1\} \}$ would
be an alternating cycle of the form we have excluded above.  
We have four subcases with respect to the existence of edges between
$v_5$ and $v_7$. 

case a) $\{v_5,v_7\} \in E(G) \setminus E(G')$:\\
This would imply the existence of the alternating cycle $C =
(v_5,v_4,v_1,v_7,v_5)$, excluded above.

case b) $\{v_5,v_7\} \in E(G') \setminus E(G)$:\\ 
This would imply the existence of the alternating cycle
$C = (v_7, v_5, \dots, v_6, v_4, v_1, v_8, \dots, v_7)$, also excluded
above.

case c)  $\{v_5,v_7\} \in E(G) \cap E(G')$:\\
Then there is an alternating cycle $C' = (v_7,v_5,v_4,v_1,v_7)$ 
on which we can swap the edges in a single step. 
This leads to a realization $G^* = G_1$ with $|G_0\Delta G^*|= 4$
and $|G^* \Delta G'| = 2\ell$.  

case d) $\{v_5,v_7\} \not\in E(G)$ and $\{v_5,v_7\} \not\in E(G')$:\\
As in case c), we consider the alternating cycle $C' =
(v_7,v_5,P_1 \setminus \{ \{v_4,v_5 \}\}, \{ v_1,v_4\}, P_2 \setminus \{
\{v_7,v_1\}\},v_7)$. Swapping edges on $C'$, we get
a realization $G^* = G_{k-1}$, but
this time, $|G_0 \Delta G^*| = |C'| \leq 2 \ell$ and $|G^* \Delta G'| 
= | G \Delta G'| - (|C'| -1) +1  \leq 2\ell$.
According to the induction hypothesis there exist sequences
$G_0^{1}:=G,G_1^{1},\dots,G_{k_1}^1:=G^*$ and
$G_0^{2}:=G^*,G_1^{2},\dots,G_{k_2}^2=G'$ with $k_1\leq
\frac{1}{2}|G_0\Delta G^*|-1 = \frac{1}{2} (| G \Delta G'| - |C'| +2)
- 1$
and $k_2 \leq \frac{1}{2} |G^* \Delta G'|-1 = \frac{1}{2} |C'| -1$. 
We arrange these sequences
one after another and get a sequence which fulfills
$k=k_1 + k_2  \leq \frac{1}{2} (| G \Delta G'| - |C'| +2)
- 1 + \frac{1}{2} |C'| -1  = \frac{1}{2} | G \Delta G'| -1$.

\item[case 4:] Assume $\{v_1,v_4\}\notin E(G)\cup E(G')$.
This implies $\{v_1,v_4\}\notin G\Delta G'.$
It exists the alternating cycle $C:=(v_1, v_2, v_3, v_4, v_1))$ with $\{v_1,v_4\}\notin E(G).$ $G_1:=(G_0\setminus \{\{v_1,v_2\},\{v_3,v_4\}\})\cup \{\{v_3,v_2\},\{v_1,v_4\}\}$ is a realization of $S$ and it follows
$|G_0\Delta G_1|=4$ and $|G_1\Delta G'|=2\ell+2-2=2\ell.$ 
According to the induction hypothesis there exist realizations
$G_1,G_2,\dots,G_k$ with $G_k:=G'$  where  and $k \leq
\frac{1}{2}|G \Delta G'|-2.$ Hence, we get the sequence
$G_0,G_1,\dots,G_k$ with $k \leq \frac{1}{2}|G\Delta G'|-1$.
\end{enumerate}
\end{proof}

We have shown that 
the state graph $\Psi = (V_{\psi}, A_{\psi})$ is a $d$-regular,
symmetric, non-bipartite, and strongly connected digraph. Hence, the
corresponding Markov chain has the uniform distribution as its
stationary distribution.
A random walk on $\Psi = (V_{\psi}, A_{\psi})$ can be described by
Algorithm~\ref{alg:undirected}. 
This algorithm requires a data structure $DS$ containing all pairs of 
non-adjacent edges in $G$.\\

\begin{algorithm}[t]
    \caption{\label{alg:undirected}Switching Algorithm} 
    \begin{algorithmic}[1]
      \REQUIRE sequence $S$, an undirected graph $G=(V,E)$ with 
$d_{G}(v_i)= a_{i}  \text{ for all } i\in \{1,\dots,n\}$ and $v_i\in V$, a mixing time $\tau.$
      \ENSURE A sampled undirected graph $G'=(V,E')$ with $d_{G}(v_i)=
      a_{i} \text{ for all } i\in \{1,\dots,n\}$ and $v_i\in V.$
      \STATE $t:=0,~G':=G$ \COMMENT{initialization}
      \WHILE{$t<\tau$}
	\STATE Choose an element $p$ from $DS$ uniformly at random.\COMMENT{$p$ is a pair of non-adjacent edges.}
        \STATE Let $p$ be the pair of edges
        $\{v_{i_{1}},v_{i_{2}}\},\{v_{i_{3}},v_{i_{4}}\}$.
        \STATE Choose with probability $\frac{1}{2}$ between case a) and case b). 
\IF{case a)}
   \IF{$\{v_{i_{1}},v_{i_{4}}\},\{v_{i_{3}},v_{i_{2}}\}\notin E(G')$}
   \STATE\COMMENT{Either walk on to an adjacent realization}
		\STATE Delete $\{v_{i_{1}},v_{i_{2}}\},\{v_{i_{3}},v_{i_{4}}\}$ in $E(G').$
		\STATE Add $\{v_{i_{1}},v_{i_{4}}\},\{v_{i_{3}},v_{i_{2}}\}$ to $E(G').$
	   \ELSE \STATE \COMMENT{or walk a loop: `Do nothing'}
		
	   \ENDIF
\ELSE \STATE \COMMENT{ case b)}
\IF{$\{v_{i_{1}},v_{i_{3}}\},\{v_{i_{2}},v_{i_{4}}\}\notin E(G')$}
   \STATE\COMMENT{Either walk on to an adjacent realization}
		\STATE Delete $\{v_{i_{1}},v_{i_{2}}\},\{v_{i_{3}},v_{i_{4}}\}$ in $E(G').$
		\STATE Add $\{v_{i_{1}},v_{i_{3}}\},\{v_{i_{2}},v_{i_{4}}\}$ to $E(G').$
	   \ELSE \STATE \COMMENT{or walk a loop: `Do nothing'}
		
	   \ENDIF

         \ENDIF
	\STATE update data structure $DS$
	\STATE $t\leftarrow t+1$
      \ENDWHILE
       \end{algorithmic}
  \end{algorithm}  

\section{Sampling Digraphs}\label{sec:general-markov}

We now turn the directed case. As before, we start with the formal
problem definition and some additional notation. Then, we introduce our
Markov chain and analyze its properties.

\paragraph{Formal problem definition}

In the directed case,
we define a degree sequence $S$ as a sequence
 of $2$-tuples $\left( {a_1 \choose b_1}, {a_2 \choose b_2}, \dots, {a_n \choose b_n}\right)$ with $a_i,b_i \in \mathbb{Z}_0^+, i
 \in \{1,\dots,n\}$ where $a_i > 0$ or $b_i>0.$

Let $G=(V,A)$ be a
 directed labeled graph $G=(V,A)$ without loops and parallel arcs and
 $|V|=n$. We define the \emph{in-degree-function} $d_G^+:V\rightarrow
 \mathbb{Z}_0^+$ which assigns to each vertex $v_i\in V$ the number of
 incoming arcs and the \emph{out-degree-function} $d_G^-:V\rightarrow
 \mathbb{Z}_0^+$ which assigns to each vertex $v_i\in V$ the number of
 outgoing arcs.   We denote $S$ as
 \emph{graphical sequence} if and only if there exists at least one
 directed labeled graph $G=(V,A)$ without any loops or parallel arcs
 which satisfies
 $d_G^+(v_i)=b_i$ and $d_G^-(v_i)=a_i$ for all $v_{i}\in
 V$ and $i\in \{1,\dots, n\}.$ Any such graph $G$ is called
 \emph{realization} of $S$. Let $H$ be a subdigraph of $G.$ We say that
 $H=(V_H,A_H)$ is an \emph{induced subdigraph} of $G$ if every arc of $A$
 with both end vertices in $V_H$ is also in $A_H.$ We write
 $H=G\left\langle V_H \right \rangle.$

The symmetric difference
$G\Delta G'$ of two realizations $G \neq G'$ is defined analogously 
to the undirected case. 
Consider for example the realizations $G$ and $G'$ with
$A(G):=\{(v_1,v_2),(v_3,v_4)\}$ and $A(G'):=\{(v_1,v_4),(v_3,v_2)\}$
consisting of exactly two arcs. Then the symmetric difference is the
alternating directed 
$4$-cycle $C:=(v_1,v_2,v_3,v_4,v_1)$ where $(v_i,v_{i+1})\in
A(G)$ for $i\in \{1,3\}$ and $(v_{i+1},v_{i})\in A(G')$ taking indices
$i\bmod{4}.$ We define an \emph{alternating directed walk} $P$ for a
directed graph $G=(V,A)$ as a sequence $P:=(v_1,v_2,\dots, v_l)$ of
vertices $v_i\in V$ where either $(v_i,v_{i+1})\in A(G)$ and
$(v_{i},v_{i-1})\notin A(G)$  or $(v_i,v_{i+1})\notin A(G)$ and
$(v_{i},v_{i-1})\in A(G)$ for $i\bmod{2}=1.$ We call an even alternating
directed walk $C$ \emph{alternating directed cycle} if $v_1=v_l$ is
fulfilled. 
The symmetric difference of two realizations always decomposes into a
number of alternating directed cycles, see
Figs.~\ref{fig:example-symm.diff} and~\ref{fig:decomposition-alternating}.

\begin{figure}[t]
\centerline{\includegraphics[height=2.5cm]{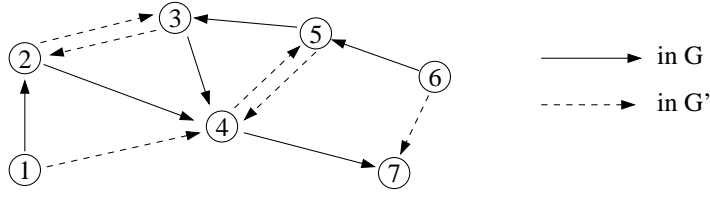}}
\caption{\label{fig:example-symm.diff}Example: Two realizations $G$ and $G'$.}
\end{figure}

\begin{figure}[t]
\centerline{\includegraphics[height=2.5cm]{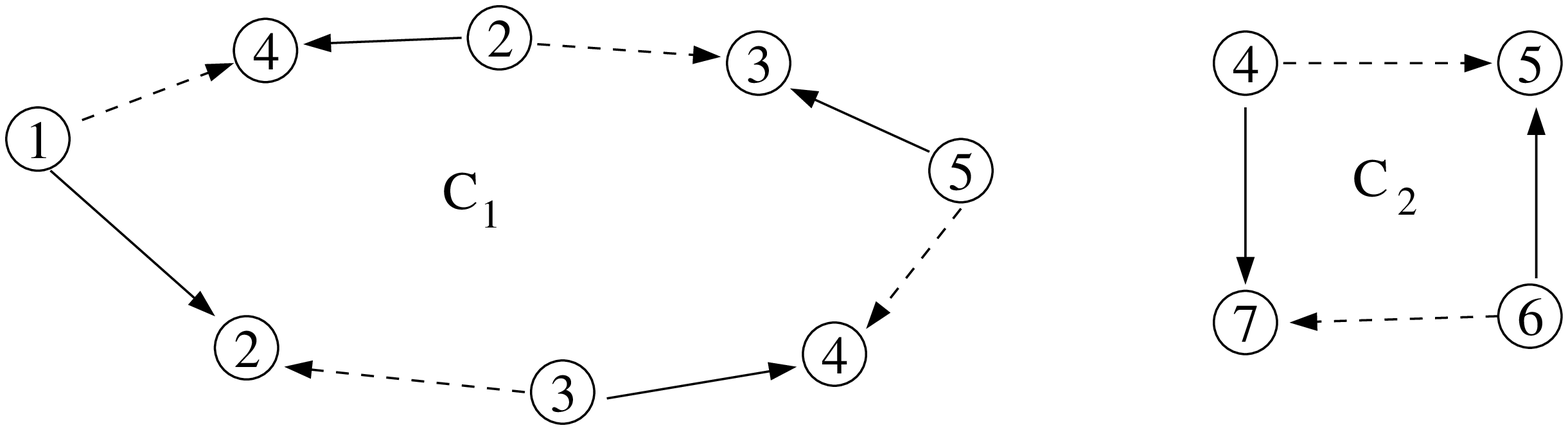}}
\caption{\label{fig:decomposition-alternating} Decomposition of the
  symmetric difference $G \Delta G'$ of
  Fig.~\ref{fig:example-symm.diff} into a minimum number of
  alternating directed cycles.}
\end{figure}

\paragraph{The Markov chain.}

In the directed case, 
we denote the state graph for our random walk by 
$\Phi = (V_{\phi}, A_{\phi})$. 
Its underlying vertex set $V_{\phi}$ is the set of
all realizations of a given degree sequence $S$.
For a realization $G$, we denote by $V_G$ the
corresponding vertex in $V_{\psi}$.
The arc set $A_{\psi}$ is defined as follows.
\begin{enumerate}
 \item[a)] We connect two vertices $V_G, V_{G'} \in V_{\phi}, G\neq G'$ with arcs $(V_G,V_{G'})$ and $(V_{G'},V_G)$ if and only if one of the two following constraints is fulfilled
\begin{enumerate}
 \item[1.] $|G\Delta G'|=4$ 
\item[2.] $|G\Delta G'|=6$ and $G\Delta G'$ contains exactly three different vertices.
\end{enumerate}
\item[b)] We set a directed loop $(V_G,V_G)$ 
\begin{enumerate}
\item[1.] for each pair of non-adjacent arcs $(v_{i_{1}},v_{i_{2}}),(v_{i_{3}},v_{i_{4}})\in A(G), i_{j}\in \{1,\dots,n\}$ if and only if $(v_{i_{1}},v_{i_{4}})\in A(G) \lor (v_{i_{3}},v_{i_{2}})\in A(G)$ in a realization $G,$
\item[2.] for each directed $2$-path $(v_{i_{1}},v_{i_{2}}),(v_{i_{2}},v_{i_{3}})\in A(G)$ if and only if one of the following constraints is true for a realization $G,$
\begin{enumerate}
 \item[i)] $(v_{i_{2}},v_{i_{1}})\in A(G) \lor (v_{i_{3}},v_{i_{2}})\in A(G) \lor (v_{i_{1}},v_{i_{3}})\in A(G)$, 
\item[ii)] $(v_{i_{3}},v_{i_{1}})\notin A(G)$,
\item[iii)] $i_{3}<i_{1}\lor i_{3}<i_{2}$.
\end{enumerate}
\item[3.] if $G$ contains no directed $2$-path.
\end{enumerate}
\end{enumerate}

\begin{Lemma}
 The state graph $\Phi:=(V_{\phi},A_{\phi})$ is non-bipartite, symmetric, and regular. 
\end{Lemma}

\begin{proof}
 In our setting we connect two vertices at each time in both
 directions. Hence, $\Phi$ is symmetric. Furthermore, if some
 realization  $G$ contains no directed $2$-path, then each $G$ is a
 realization of a sequence~$S$, only consisting of sinks and
 sources. With our setting $\Phi$ contains for each $V_G\in V_{\phi}$ a
 directed loop and is therefore non-bipartite, see item $b) 3$ in our
 construction. Let us now assume that a realization $G$ contains a
 directed $2$-path. Either there exists a third arc which completes
 these two arcs to a directed $3$-cycle or not. In all cases we can
 guarantee one directed loop at $V_G:$ In the case of a directed
 $3$-cycle $C$ we distinguish two cases. Either $b)2.i)$ is fulfilled
 or in $C$ there exists a $2$-path with conditions as in $b)
 2.iii)$. If we have a $2$-path which is not a subpath of a directed
 $3$-cycle then we get condition $b)2.ii).$ Hence, $\Phi$ is not
 bipartite. For the proof of regularity, note, that we consider at
 each vertex $V_G$ the number of pairs of non-adjacent arcs in a
 realization $G.$ This is the number of all possible arc pairs minus
 the number of adjacent arcs ${|A(G)|\choose 2}-(\sum_{i=1}^{n}{a_i
   \choose 2}+\sum_{i=1}^{n}{b_i \choose 2}+\sum_{i=1}^{n}a_ib_i)$
 where $\sum_{i=1}^{n}{a_i \choose 2}$ is the number of all incoming
 arc pairs at each vertex, $\sum_{i=1}^{n}{b_i \choose 2}$ is the
 number of all outgoing arc pairs at each vertex and
 $\sum_{i=1}^{n}a_ib_i$ is the number of directed $2$-paths in a
 realization $G.$ Hence, the number of non-adjacent arcs is a constant
 value for each realization $G.$ For each of these arc pairs we either
 set a directed loop or an incoming and an outgoing arc at each vertex
 $V_G\in V_{\phi}$. For each $2$-path in $G$ we set a loop if it is not
 part of a directed $3$-cycle
 $C=(v_{i_{1}},v_{i_{2}},v_{i_{3}},v_{i_{1}})$ which is an induced
 subdigraph $C=G\left\langle
   \{v_{i_{1}},v_{i_{2}},v_{i_{3}}\}\right\rangle$. If it is the
 case it exists a realization $G'$ with $|G\Delta G'|=6$ and $G\Delta
 G'$ contains exactly $3$ different vertices. Hence, we set for the
 $2$-path in $C$ with $i_j<i_{j'}$ and $i_{j'}<i_{j''}$ with $j,j',j''
 \in \{1,2,3\}$ the directed arcs $(V_G,V_{G^{'}})$ and $(V_{G'},V_G)$
 and for both other $2$-paths in $C$ a directed loop. Generally, we
 set for all $2$-paths in a realization an incoming and an outgoing
 arc at each $V_G.$ The number of $2$-paths in each realization is the
 constant value $\sum_{i=1}^{n}a_i b_i.$ Hence, the vertex degree at
 each vertex is 
$d_{\Phi}:=d^+_{\Phi} = d^-_{\Phi}= {|A(G)|\choose 2}-2\sum_{i=1}^{n}{a_i \choose 2}.$
\end{proof}

In the next section we have to prove that our constructed graphs are strongly connected. This is sufficient to prove the reachability of each realization independent of the starting realization. 
Fig.~\ref{fig:swap-sequence-example} shows an example how the realization
$G$ from Fig.~\ref{fig:example-symm.diff}
can be transformed to the realization $G'$ by a sequence of swap operations.

\begin{figure}[t]
\centerline{\includegraphics[height=9cm]{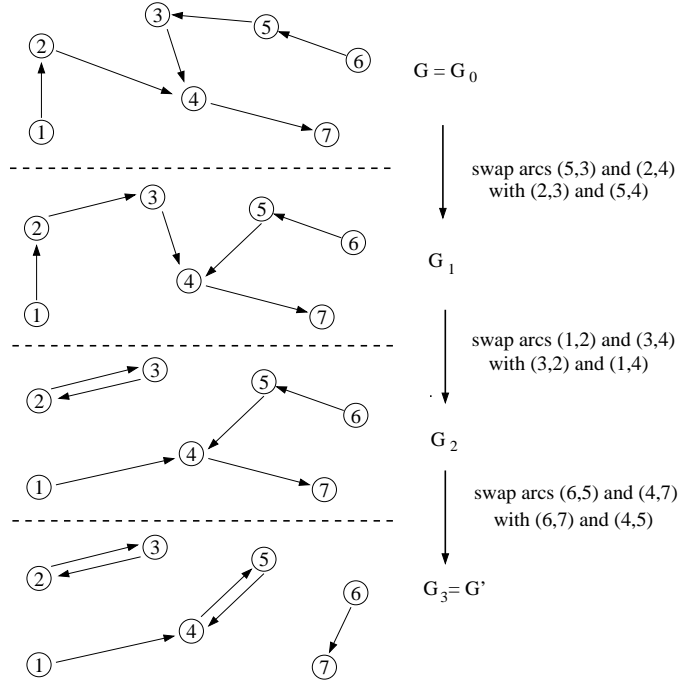}}
\caption{\label{fig:swap-sequence-example} Transforming $G$ from
  Fig.~\ref{fig:example-symm.diff} into $G'$ by a sequence of swap operations.}
\end{figure}

\subsection{Symmetric differences of two different realizations}

\begin{Proposition}\label{PR2}
Let $S$ be a graphical sequence and $G$ and $G'$ be two different
realizations. If $G\Delta G'$ is exactly one weak component and
$|G\Delta G'|\neq 6$ then there exists in $G\Delta G'$ a
vertex-disjoint alternating 3-walk of type $P$ or $Q$, where 
$P=(v_1,v_2,v_3,v_4)$ with
$(v_1,v_2),(v_3,v_4)\in A(G)$ and $(v_3,v_2)\in A(G')$
and  $Q=(w_1,w_2,w_3,w_4)$ with
$(w_1,w_2),(w_3,w_4)\in A(G')$ and $(w_3,w_2)\in A(G)$.
\end{Proposition}

\begin{proof} Note that in $G\Delta G'$ an alternating cycle of length
  two is not possible. Otherwise, there exists an arc $(u,v)\in
  A(G)\cap A(G')$ in contradiction to our assumption that $(u,v)\in
  G\Delta G'.$
The symmetric difference $G\Delta G'$ may decompose into a number of 
alternating cycles $(G \Delta G')_i$. We consider a decomposition into
the minimum number of such cycles. 
If one of these alternating cycles $(G \Delta G')_i$ 
contains a vertex-disjoint alternating $3$--walk $P$ or $Q$ as
claimed, we are done.
Otherwise, each vertex is repeated at each third
step in $(G\Delta G')_i.$ Hence, we get the alternating cycles
$(G\Delta
G')_i:=(v_{i_1},v_{i_2},v_{i_3},v_{i_1},v_{i_2},v_{i_3},v_{i_1})$ 
where
$(v_{i_1},v_{i_2}),(v_{i_2},v_{i_3}),(v_{i_3},v_{i_1})\in A(G)$ and
$(v_{i_2},v_{i_1}),(v_{i_3},v_{i_2}),(v_{i_1},v_{i_3})\in A(G').$ 
The cycle cannot be longer, as the graph induced by $G \Delta
G'\langle \{ v_1,v_2,v_3 \} \rangle$ is already complete.
Since $|(G\Delta G')_i| = 6$, there must be 
$(G\Delta G')_j$ 
with
$i \neq j$.
$(G\Delta G')_j$ shares at least one vertex with $(G\Delta G')_i$,
because $G \Delta G'$ is
weakly connected.
There must be exactly one $v_{i_1} = v_{j_1}$, since otherwise these
two cycles were not arc-disjoint.
The union of these two cycles is an alternating cycle, in
contradiction to the minimality of the decomposition.
\end{proof}

Note that the above proposition does not assert that
the symmetric difference contains $P$ and $Q$.
The smallest counter-example are the realizations
$G=(V,A)$ and $G'=(V,A')$ with $V = \{ v_1,v_2,v_3,v_4\}$ and
$A = \{ (v_1,v_3), (v_3,v_2), (v_2,v_4), (v_4,v_1) \}$ and
$A' = \{ (v_1,v_2), (v_2,v_1), (v_3,v_4), (v_4,v_3) \}$.

\begin{Proposition}\label{PR1} Let $S$ be a graphical sequence and $G$ and $G'$ be two different realizations. If  $|G\Delta G'|=6$, then there exist
\begin{enumerate}
\item[a)] realizations $G_0,G_1,G_2$ with $G_0:=G$, $G_2:=G'$ and 
          $|G_i\Delta G_{i+1}|=4$ for $i\in\{0,1\}$ or
\item[b)] $G$ and $G'$ are different in the orientation of exactly one directed $3$-cycle.
\end{enumerate}
\end{Proposition}

\begin{proof} First observe that the symmetric difference is weakly
  connected whenever $|G\Delta G'|=6$. 
We consider the alternating $6$-cycle $C:=G\Delta G'.$
\begin{itemize}
 \item[case 1:] $C$ contains at least four different vertices.
Assume first that $C$ contains four different vertices. The only possibility to realize this scenario is $C=(v_1,v_2,v_3,v_1,v_4,v_3,v_1)$ with $(v_1,v_2),(v_3,v_1),(v_4,v_3)\in A(G)$ and $(v_3,v_2),(v_4,v_1),(v_1,v_3)\in A(G')$. (A permutation of $\{1,2,3\}$ does not influence the result.) We get the alternating vertex-disjoint walk $P=(v_4,v_3,v_1,v_2).$ 
\begin{itemize}
\item[(i):] Assume $(v_4,v_2)\notin A(G).$ It follows $(v_4,v_2)\notin A(G').$ Otherwise, we would get $(v_4,v_2)\in G\Delta G'$ in contradiction to our assumption. We set
$$G_1:=(G_0\setminus \{(v_4,v_3),(v_1,v_2)\})\cup \{(v_4,v_2),(v_1,v_3)\}~\textnormal{and}$$
$$G_2:=(G_1\setminus \{(v_4,v_2),(v_3,v_1)\})\cup \{(v_4,v_1),(v_3,v_2)\}.$$
We get $G_2=G'$ and realizations $G_0,G_1,G_2$ with $|G_i\Delta G'_{i+1}|=4.$
\item[(ii):] Assume $(v_4,v_2)\in A(G).$ It follows $(v_4,v_2)\in A(G').$ Otherwise, we would get $(v_4,v_2)\in G\Delta G'$ in contradiction to our assumption. We set
$$G_1:=(G_0\setminus \{(v_4,v_2),(v_3,v_1)\})\cup \{(v_4,v_1),(v_3,v_2)\}~\textnormal{and}$$
$$G_2:=(G_1\setminus \{(v_4,v_3),(v_1,v_2)\})\cup \{(v_4,v_2),(v_1,v_3)\}.$$
We get $G_2=G'$ and realizations $G_0,G_1,G_2$ with $|G_i\Delta G'_{i+1}|=4.$
\end{itemize}
We can argue analogously if $C$ contains five our six different vertices.
\item[case 2:] $C$ contains exactly three different vertices. Then $C$
  is the alternating cycle $C=(v_1,v_2,v_3,v_1,v_2,v_3,v_1)$ with
  $(v_1,v_2),(v_2,v_3),(v_3,v_1) \in A(G)$ and
  $(v_3,v_2),(v_2,v_1),(v_1,v_3) \in A(G').$ Hence, $G$ and $G'$ are different in the orientation of exactly one directed $3$-cycle.
\end{itemize}
\end{proof}


\begin{Lemma}\label{TH1} Let $S$ be a graphical sequence and $G$ and
  $G'$ be two different realizations. There exist 
realizations $G_0,G_1,\dots,G_k$ with $G_0:=G$, $G_k:=G'$ and 
\begin{enumerate}
 \item $|G_i\Delta G_{i+1}|=4$ or
\item $|G_i\Delta G_{i+1}|=6$
\end{enumerate}
where $k\leq \frac{1}{2}|G\Delta G'|-1.$ In case $(2),$ $G_i\Delta G_{i+1}$ consists of a directed $3$-cycle and its opposite orientation.
\end{Lemma}

\begin{proof} 
We prove the lemma by induction according to the cardinality of the
symmetric difference $|G\Delta G'|=2\kappa.$ For $\kappa:=2$ we get
$|G\Delta G'|=4.$ The correctness of our claim follows with
$G_1:=G'$. For $\kappa:=3$ we get a sequence of realizations
$G_0,G_1,G_2$ with case $a)$ of Proposition \ref{PR1}. In case $b)$ we
get a directed $3$-cycle with its opposite orientation. In both cases
it follows $k\leq 2.$ 

We assume the correctness of our claim for all $\kappa\leq \ell.$ Let
$|G\Delta G'|=2\ell+2.$ We can assume that $\kappa>3.$ 
Assume further, that the symmetric difference consists of $k$ weakly
connected components $(G \Delta G')_i$ for $i \in \{1, \dots, k\}$.

Consider first the case that
for all these components  $|(G \Delta G')_i| = 6$ and that each
component contains exactly three distinct vertices, then each of them
is a directed 3-cycle and its reorientation. We choose  $(G \Delta
G')_1$, perform a 3-cycle reorientation on it, and obtain realization
$G^*$. Thus $|G^* \Delta G'| = 2\ell -4$. By the induction hypothesis,
there are realizations $G_0 = G^*, G_1, \dots, G_k = G'$ such that
$ k \leq \frac{1}{2}|G^*\Delta G'|-1 < \frac{1}{2}|G\Delta G'|-1$. 
Combining the first 
3-cycle reorientation with this sequence of realizations gives the
desired bound. 
If there is  a component $|(G \Delta G')_i| = 6$ with at least
four distinct vertices, we can apply Proposition~\ref{PR1}, case a)
to it and handle the remaining components by induction.

Otherwise, there is a component with $|(G \Delta G')_i| \geq 8$. 
Due to Proposition~\ref{PR2}, we may assume that 
there is a vertex-disjoint walk
$P=(v_1,v_2,v_3,v_4)$ with $(v_1,v_2),(v_3,v_4)\in A(G)$ and
$(v_3,v_2)\in A(G').$ Otherwise, there exists 
$Q=(w_1,w_2,w_3,w_4)$ 
with $(w_1,w_2),(w_3,w_4)\in A(G')$ and $(w_3,w_2)\in A(G)$.
In that case we can exchange the roles of $G$ and $G'$ 
and consider $G' \Delta G$. Clearly, a sequence of realizations
$G' = G'_0, G'_1, \dots, G'_k = G$ can be reversed
and then fulfills the conditions of the lemma.
So from now on we work with $P$.
 
\begin{enumerate}
 \item[case 1:] Assume $(v_1,v_4)\in A(G')\setminus A(G).$\\
This implies $(v_1,v_4)\in G\Delta G'$.
$G_1:=(G_0\setminus \{(v_1,v_2),(v_3,v_4)\})\cup \{(v_3,v_2),(v_1,v_4)\}$ is a realization of $S$ and it follows
$|G_0\Delta G_1|=4$ and $|G_1\Delta G'|=2\ell+2-4=2(\ell-1).$ 
Note that after this step, $G_1\Delta G'$ may consist of several
connected components, but each of them has strictly smaller
cardinality. Therefore, we can apply the induction hypothesis on $|G_1\Delta G'|$. 
Thus, we obtain realizations $G_1,G_2,\dots,G_k$ with $G_k:=G'$ and $|G_i \Delta G_{i+1}|=4$ or $|G_i \Delta G_{i+1}|=6$ where $k-1\leq \frac{1}{2}|G_1\Delta G'|-1.$ Hence, we get the sequence $G_0,G_1,\dots,G_k$ with $k= 1+\frac{1}{2}|G_1\Delta G'|-1=\frac{1}{2}(|G\Delta G'|-4)\leq \frac{1}{2}|G\Delta G'|-1$ which fulfills $1.$ and $2.$
\item[case 2:] Assume $(v_1,v_4)\in A(G)\cap A(G').$\\
This implies $(v_1,v_4)\notin G\Delta G'.$
Consider an alternating cycle $C=(v_4,v_i,\dots,v_j,$
$v_1,v_2,v_3,v_4)$ of $(G \Delta G')_i$ such that
each vertex has in-degree two or out-degree two. 
Then $P$ is an alternating subpath of
$C$ with $(v_i,v_4),(v_1,v_j)\in A(G')$. We construct a new
alternating cycle $C^*:=(C\setminus P)\cup \{(v_1,v_4)\}$
with length $|C^*|=|C|-2.$ 
We swap the arcs in $C^*$ and get a realization $G^*$ of S with
$|G_0\Delta G^*|=|C^*|\leq 2\ell$ and $|G^*\Delta G'|=|G \Delta
G'|-(|C^*|-1)+1\leq 2\ell.$ 
According to the induction hypothesis there exist sequences
$G_0^{1}:=G,G_1^{1},\dots,G_{k_1}^1:=G^*$ and
$G_0^{2}:=G^*,G_1^{2},\dots,G_{k_2}^2=G'$ with $k_1\leq
\frac{1}{2}|G_0\Delta G^*|-1$ and $k_2\leq \frac{1}{2} | G^* \Delta
G'| -1$. We arrange these sequences one after another and get a
sequence which fulfills $1.$ and $2.$ and $k=k_1+k_2=\frac{1}{2}|G_0
\Delta G^*|-1 + \frac{1}{2} | G^* \Delta G'| -1
=\frac{1}{2}|G \Delta G'|-1$.
\item[case 3:] Assume $(v_1,v_4)\in A(G)\setminus A(G').$\\
This implies $(v_1,v_4)\in G\Delta G'.$
The alternating walk $P$ can be extended to an alternating cycle $C =
\{ v_1, v_2, v_3, v_4, v_5, \dots, v_{2t}, v_1\}, t\geq 3$ using
only arcs from $G\Delta G'$. To construct $C$, start with $P$, and keep
adding alternating arcs until you reach the start vertex $v_1$ for the
first time. Obviously, you will not
get stuck before reaching $v_1$. 
Note that the arc $(v_1,v_4)$
does not belong to $C$. Therefore, there exists 
an alternating sub-cycle $C^*:=C \cup \{ (v_1,v_4)\} \setminus P$ 
formed by arcs in $G\Delta G'$. 
We swap the arcs in $C^*$ and get a realization $G^*$ of S with 
$|G_0\Delta G^*|=|C^*|\leq 2\ell$ and $|G^*\Delta G'|=|G \Delta
G'|-|C^*|\leq 2\ell$. According to the induction hypothesis there
exist sequences $G_0^{1}:=G,G_1^{1},\dots,G_{k_1}^1:=G^*$ and
$G_0^{2}:=G^*,G_1^{2},\dots,G_{k_2}^2=G'$ with $k_1\leq
\frac{1}{2}|G_0\Delta G^*|-1$ and $k_2\leq \frac{1}{2} (|G \Delta G'|
- |C^*|)-1 $. We arrange these sequences
one after another and get a sequence which fulfills $1.$ and $2.$ 
and $k=k_1+k_2=\frac{1}{2}|G_0\Delta G^*|-1+ \frac{1}{2} (|G \Delta G'|
- |C^*|)-1 =\frac{1}{2}(|G\Delta G'|) -2$.

\item[case 4:] Assume $(v_1,v_4)\notin A(G)\cup A(G')$.
This implies $(v_1,v_4)\notin G\Delta G'.$
It exists the alternating cycle $C:=(P,(v_1,v_4))$ with $(v_1,v_4)\notin A(G).$ $G_1:=(G_0\setminus \{(v_1,v_2),(v_3,v_4)\})\cup \{(v_3,v_2),(v_1,v_4)\}$ is a realization of $S$ and it follows
$|G_0\Delta G_1|=4$ and $|G_1\Delta G'|=2\ell+2-2=2\ell.$ 
According to the induction hypothesis there exist realizations $G_1,G_2,\dots,G_k$ with $G_k:=G'$ which fulfill $1.)$ and $2.)$ where $k_1:=1$ and $k_2:=k-1\leq \frac{1}{2}|G_1\Delta G'|-1.$ Hence, we get the sequence $G_0,G_1,\dots,G_k$ with $k=k_1+k_2= 1+\frac{1}{2}|G_1\Delta G'|-1=\frac{1}{2}(|G\Delta G'|-3+1)= \frac{1}{2}|G\Delta G'|-1$ which fulfills $1.$ and $2.$
\end{enumerate}
\end{proof}

\begin{Corollary}
 State graph $\Phi$ is a strongly connected directed graph.
\end{Corollary}

\subsection{Random Walks}

A random walk on $\Phi = (V_{\phi}, A_{\phi})$ can be described by
Algorithm~\ref{alg:directed}. 
We now require a data structure $DS$ containing all pairs of
non-adjacent arcs and all directed $2$-paths in the current
realization.

\begin{algorithm}[h]
    \caption{\label{alg:directed}Sampling realization digraphs} 
    \begin{algorithmic}[1]
      \REQUIRE sequence $S$, a directed graph $G=(V,A)$ with ${d^{+}_{G}(v_i) \choose d^{-}_{G}(v_i)}={a_{i}\choose b_{i}}~ \forall i\in \{1,\dots,n\}$ and $v_i\in V$, a mixing time $\tau.$
      \ENSURE A sampled directed graph $G'=(V,A')$ with ${d^+_{G'}(v_i) \choose d^-_{G'}(v_i)}={a_{v_{i}} \choose b_{v_{i}}} ~\forall i\in \{1,\dots,n\}$ and $v_i\in V.$
      \STATE $t:=0,~G':=G$ \COMMENT{initialization}
      \WHILE{$t<\tau$}
	\STATE Choose an element $p$ from $DS$ uniformly at random.\COMMENT{$p$ is a pair of non-adjacent arcs or a directed $2$-path.}
        \IF{$p$ is a pair of non-adjacent arcs $(v_{i_{1}},v_{i_{2}}),(v_{i_{3}},v_{i_{4}})$}
          \IF{$(v_{i_{1}},v_{i_{4}}),(v_{i_{3}},v_{i_{2}})\notin A(G')$}
		\STATE\COMMENT{Either walk on $\Phi$ to an adjacent realization $G'$}
		\STATE Delete $(v_{i_{1}},v_{i_{2}}),(v_{i_{3}},v_{i_{4}})$ in $A(G').$
		\STATE Add $(v_{i_{1}},v_{i_{4}}),(v_{i_{3}},v_{i_{2}})$ to $A(G').$
	   \ELSE \STATE \COMMENT{or walk a loop: `Do nothing'}
		
	   \ENDIF
	       \ELSE \STATE \COMMENT{$p$ is a directed $2$-path $P=(v_{i_{1}},v_{i_{2}},v_{i_{3}})$}
          \IF{$((v_{i_{3}},v_{i_{1}})\in A(G'))\land
            ((v_{i_{2}},v_{i_{1}}),(v_{i_{3}},v_{i_{2}}),(v_{i_{1}},v_{i_{3}})\notin A(G')) \land (i_3 > i_1) \land (i_3 > i_2) $}
		\STATE \COMMENT{Walk on $\Phi$ to an adjacent
                  realization $G'$ with a reoriented directed $3$-cycle}
          	\STATE Delete $(v_{i_{1}},v_{i_{2}}),(v_{i_{2}},v_{i_{3}}),(v_{i_{3}},v_{i_{1}})$ in $A(G').$
	  	\STATE Add $(v_{i_{2}},v_{i_{1}}),(v_{i_{3}},v_{i_{2}}),(v_{i_{1}},v_{i_{3}})$ to $A(G').$
	   \ELSE \STATE \COMMENT{Walk a loop: `Do nothing'}
	
         \ENDIF
	\ENDIF
	\STATE update data structure $DS$
	\STATE $t\leftarrow t+1$
      \ENDWHILE
       \end{algorithmic}
  \end{algorithm}  

\begin{Theorem}
 Algorithm \ref{alg:directed} is a random walk on state graph $\Phi$
 which samples uniformly at random a directed graph $G'=(V,A)$ as a realization of sequence $S$ for $\tau\rightarrow \infty.$  
\end{Theorem}

\begin{proof}
Algorithm \ref{alg:directed} chooses elements in $DS$ with the same
constant probability. For a vertex $V_G\in V_{\phi}$ there exist for
all these pairs of arcs in $A(G')$ either incoming and outgoing arcs
on $V_{G'}$ in $\Phi$ or a loop. Let $d_{\phi}:={|A(G)|\choose
  2}-2\sum_{i=1}^{n}{a_i \choose 2}.$ We get a transition matrix $M$
for $\Phi$ with $p_{ij}=\frac{1}{d_{\phi}}$ for $i,j\in A(\Phi),i\neq j$,
$p_{ij}=1-\sum_{\{i|(i,j)\in A(\Phi), ~i\neq j\}}\frac{1}{d_{\phi}}$ for
$i,j\in V_{\phi},i=j$, otherwise we set $p_{ij}=0.$ Since, $\Phi$ is a
regular, strongly connected, symmetrical and non-bipartite directed
graph, the distribution of all realizations in a $t$th step converges
asymptotically to the uniform distribution.  
\end{proof}

\section{Arc-Swap Sequences}\label{sec:arc-swap-sequences}

In this section, we study under which conditions the simple switching
algorithm works correctly for digraphs. The Markov chain used in the switching
algorithm works on the following simpler state 
graph $\overline{\Phi}=(V_{\overline{\phi}}, A_{\overline{\phi}})$. 
We define $A_{\overline{\phi}}$ as follows.
\begin{enumerate}
\item[a)] We connect two vertices $V_G,V_{G'}\in V_{\overline{\phi}}, G\neq G'$ with arcs $(V_G,V_{G'})$ and $(V_{G'},V_G)$ if and only if $|G\Delta G'|=4$ is fulfilled.
\item[b)] We set for each pair of non-adjacent arcs $(v_{i_{1}},v_{i_{2}}),(v_{i_{3}},v_{i_{4}})\in A(G), i_{j}\in \{1,\dots,n\}$ a directed loop $(V_G,V_G)$ if and only if 
 $(v_{i_{1}},v_{i_{4}})\in A(G) \lor (v_{i_{3}},v_{i_{2}})\in A(G)$.
\item[c)] We set one directed loop $(V_G,V_G)$ for all $V_G\in V_{\overline{\phi}}.$ 
\end{enumerate}

\begin{Lemma}
 The state digraph $\overline{\Phi}=(V_{\overline{\phi}}, A_{\overline{\phi}})$ is non-bipartite, symmetric, and regular. 
\end{Lemma}

\begin{proof} Since each vertex $V_G\in V_{\overline{\phi}}$ contains
  a loop, 
$\overline{\Phi}$ is not bipartite. At each time we set an arc we also
do this for its opposite direction. Hence, $\overline{\Phi}$ is
symmetric. The number of incoming and outgoing arcs at each $V_G$
equals the number of non-adjacent arcs in $G$, which is the constant
value ${|A(G)|\choose 2}-\left(\sum_{i=1}^{n}{a_i \choose
    2}+\sum_{i=1}^{n}{b_i \choose 2}+\sum_{i=1}^{n}a_ib_i \right).$ 
Thus, we get the regularity of $\overline{\Phi}.$
\end{proof}

\subsection{Characterization of Arc-Swap Sequences}

As shown in Example~\ref{example:disconnectedness}
in the Introduction, 
$\overline{\Phi}$ decomposes into several components, but we are able to
characterize sequences $S$ for which strong connectivity is fulfilled
in $\overline{\Phi}$. 
In fact, we will show that there are numerous sequences which only
require switching by 2-swaps. In the following we give necessary and sufficient conditions allowing to identify such sequences in polynomial running time.

\begin{Definition}\label{Def:1}
Let $S$ be a graphical sequence and let $G=(V,A)$ be an arbitrary realization. We denote a vertex subset $V'\subseteq V$ with $|V'|=3$ as an \emph{induced cycle set} $V'$ if and only if for each realization $G^*=(V,A^*)$ the induced subdigraph $G^*\left\langle V'\right \rangle$ is a directed $3$-cycle.
\end{Definition}

\begin{Definition}\label{Def:2}
Let $S$ be a graphical sequence and $G=(V,A)$ an arbitrary realization. We call $S$ an \emph{arc-swap-sequence} if and only if each subset $V'\subseteq V$ of vertices with $|V'|=3$ is not an induced cycle set.
\end{Definition}

This definition enables us  to use a simpler state graph for sampling
a realization $G$ for arc-swap-sequences. In
Theorem~\ref{th:arc-swap-main}, we will show show that in these cases 
we have only to
switch the ends of two non-adjacent arcs.

Before, we study how to recognize arc-swap sequences efficiently.
 Clearly, we may not determine all realizations to identify a sequence
 as an arc-swap-sequence. Fortunately, we are able to give a
 characterization of sequences allowing us to identify an
 arc-swap-sequence in only considering one realized digraph. We need a
 further definition for a special case of symmetric differences.

\begin{Definition}\label{def:simple-symmetric-difference}
Let $S$ be a graphical sequence and $G=(V,A)$ and $G^*=(V,A^*)$
arbitrary realizations. We call $G\Delta G^*$ \emph{simple symmetric
  cycle} if and only if each vertex $v \in V(G\Delta G^*)$
possesses vertex in-degree $d^{-}_{G\Delta G^*}(v)\leq 2$ and vertex
out-degree $d^{+}_{G\Delta(v) G^*}\leq 2$, and if $G\Delta G^*$ is an
alternating directed cycle. 
\end{Definition}

Note that the alternating directed cycle $C_1$ 
in Fig.~\ref{fig:decomposition-alternating}
is not a simple symmetric cycle, because $d^+_{C_1}(4) = 4$. 
Cycle $C_1$ decomposes into two simple symmetric cycles $C'_1 = \{
v_1,v_2,v_3,v_4,v_1 \}$ and $C''_1 = \{ v_2, v_3, v_5, v_4, v_2 \}$.

\begin{Theorem}\label{th:characterizing-arc-swap-sequences}
A graphical sequence $S$ is an arc-swap-sequence if and only if for any realization $G=(V,A)$ the following property is true: \\
For each induced, directed $3$-cycle $G\left \langle V'\right \rangle$
of $G$ there exists a realization $G^*=(V,A^*)$ 
so that $G \Delta G^*$ is
a simple symmetric cycle and that the induced subdigraph $G^*\left \langle V'\right \rangle$ is not a directed $3$-cycle.
\end{Theorem}

\begin{proof}
$\Rightarrow:$ Let $S$ be a graphical arc-swap sequence and $G=(V,A)$
be an arbitrary realization. With Definition \ref{Def:2} it follows that
each subset $V'\subset V$ with $|V'|=3$ is not an induced cycle
set. Hence, there exists for each induced, directed $3$-cycle $G\left
  \langle V'\right \rangle$ of $G$ a realization $G'=(V,A')$ with
symmetric difference $G\Delta G'$ where the induced subdigraph
$G'\left \langle V'\right \rangle$ is not a directed cycle. If the
symmetric difference $G\Delta G'$ is not a simple symmetric cycle
we delete as long alternating cycles in $G\Delta G'$ as we get an
directed alternating cycle $C^*$ where each vertex in $C^*$ has at most
vertex in-degree two and at most vertex out-degree two. 
Furthermore, $C^*$ shall contain at least one arc $(v,v')\in V'\times
V'.$ This is possible, because $G\Delta G'$ contains at least one such 
arc. On the other hand the alternating cycle $C^*$ does not contain
all possible six of such arcs. Otherwise, the induced subdigraph
$G'\left \langle V'\right \rangle$ is a directed cycle. Now, we
construct the realization $G^*=(V,A^*)$ with $A^*:=(A(G)\setminus
(A(C^*)\cap A(G)))\cup (A(C^*)\cap A(G')).$ It follows $G\Delta
G^*=C^*$ is a simple symmetric 
difference.\\[1ex] 
$\Leftarrow:$ Let $G$ be any realization of sequence $S.$ We only have
to consider $3$-tuples of vertices $V'$ inducing directed $3$-cycles
in $G.$ With our assumption there exists for each $V'$ a realization
$G^*$ so that $G^*\left \langle V'\right \rangle$ is not a directed
$3$-cycle. Hence, we find for each subset $V'\subset V$ of vertices
with $|V'|=3$ a realization $G^*=(V,A)$, so that the induced
subdigraph $G^*\left \langle V'\right \rangle$ is not a directed
$3$-cycle. We conclude that $S$ is an 
arc-swap sequence.
\end{proof}

This characterization allows us to give a simple polynomial-time
algorithm to recognize arc-swap-sequences. All we have to do is to
check for each induced 3-cycle of the given realization, if it forms
an induced cycle set.
Therefore, we 
check for each arc $(v,w)$ in an induced 3-cycle whether  
there is an
alternating walk from $v$ to $w$ (not using arc $(v,w)$) which 
does not include all five remaining arcs of the 3-cycle and its reorientation. 
Moreover, each node on this
walk has at most in-degree 2 and at most out-degree 2.
Such an alternating walk can be found in linear time by 
using a reduction to an $f$-factor problem in a bipartite graph.
In this graph we search for an undirected alternating path by
growing
alternating trees (similar to matching algorithms in bipartite graphs,
no complications with blossoms will occur), see for example 
\cite{Schrijver03}. 
The trick to ensure that not all five arcs will
appear in the alternating cycle is to iterate over these five arcs and
exclude exactly one of them from the alternating
path search between $v$ and $w$. Of course, this loop stops
as soon as one alternating path is found. Otherwise, no such
alternating path exists.
As mentioned in the Introduction, a linear-time recognition 
is possible with a parallel Havel-Hakimi algorithm of LaMar~\cite{LaMar09}.

Next, we are going to prove that $\overline{\Phi}$ is strongly connected for 
arc-swap-sequences. The structure of the proof is similar to the case
of $\Phi$, but technically slightly more involved.

\begin{Lemma}\label{lemma:arc-swap-simple-symmetric}
Let $S$ be a graphical
   arc-swap-sequence and $G$ and $G^*$ be two different realizations. 
Assume that $V' := \{ v_1,v_2,v_3\} \subseteq V$  such that
$G\langle V'\rangle$ is an induced directed 3-cycle but
$G^*\langle V'\rangle$ is not an induced directed 3-cycle.
Moreover, assume that $G \Delta G^*$ is a simple symmetric cycle.
Then there are realizations $G_0,G_1,\dots,G_k$ with
 $G_0:=G,G_k:=G^*$, $|G_i\Delta G_{i+1}|=4$ and 
$k \leq \frac{1}{2} |G \Delta G^*|$.
\end{Lemma}

\begin{proof}
We prove this lemma by induction on the cardinality of $G \Delta G^*$.
The base case $ |G \Delta G^*|=4$ is trivial. 
Consider next the case $ |G \Delta G^*|=6$. We distinguish between two
subcases.
\begin{enumerate}
\item[case a)] $G \Delta G^*$ consists of at least four different
 vertices.\\
By Proposition~\ref{PR1}, case a), there are realizations $G=G_0,G_1,
G_2=G^*$ with $|G_i\Delta G_{i+1}|=4$.
\item[case b)] $G \Delta G^*$ consists of exactly three vertices 
$v_4,v_5,v_6$.\\
Observe that $G \Delta G^*$ contains at least one arc from $G\langle
V' \rangle$ or its reorientation but not all three vertices $V'$
as otherwise $G^*\langle V' \rangle$ would be an induced $3$-cycle.
In fact, it turns out that $G \Delta G^*$ contains exactly 
one arc, say $(v_2,v_3)$, from $G\langle V' \rangle$ and its opposite arc
$(v_3,v_2)$, because  $G \Delta G^*$ is the directed alternating 
cycle $C :=(v_2,v_3,v_4,v_2,v_3,v_4,v_2)$ with $v_4 \neq v_1$.
We have two subcases. Assume first that $(v_1,v_4) \not\in A(G) \cap
A(G^*)$. So we can swap the directed alternating cycle 
$(v_1,v_2,v_3,v_4,v_1)$ in a single step. We then obtain the directed
alternating $6$-cycle $(v_2,v_3,v_4,v_2,v_1,v_4,v_2)$ which consists
of four different vertices. By Proposition~\ref{PR1}, case a), we can
swap the arcs of this cycle in two steps, thus in total in three steps
as claimed.
Otherwise, $(v_1,v_4) \in A(G) \cap A(G^*)$. Then we obtain the
directed alternating cycle $(v_1,v_4,v_2,v_3,v_1)$ which can be
swapped  in a single step. By that, we obtain a 
new cycle $(v_1,v_3,v_4,v_2,v_3,v_4,v_1)$ 
which consists
of four different vertices. By Proposition~\ref{PR1}, case a), we can
swap the arcs of this cycle in two steps, thus in total in three steps
as claimed.
\end{enumerate}

For the induction step, let us consider $|G \Delta G^*| = 2 \ell + 2
\geq 8$.
Then $G \Delta G^*$ contains between one and five arcs 
from $G\langle V' \rangle$ and its reorientation.
By Proposition~\ref{PR2}, there is a vertex-disjoint alternating
directed walk $P = (w_1,w_2,w_3,w_4)$ in $G \Delta G^*$ 
with $(w_1,w_2) \in A(G)\setminus A(G^*)$
or $(w_1,w_2) \in A(G^*)\setminus A(G)$.

Suppose that $P$ contains no arc  from $G\langle V' \rangle$ 
and its reorientation. 
We consider the case $(w_1,w_2) \in A(G)\setminus A(G^*)$.
If $(w_1,w_4) \in A(G) \cap A(G^*)$, then we consider
$C = (G \Delta G^*) \cup \{ (w_1,w_4)\} \setminus P$. 
We swap the arcs of $C$ and obtain as realization $G^{**}$.
Clearly, $G \Delta G^{**}$ contains an arc from $G\langle V' \rangle$ 
or its reorientation, and is a simple symmetric cycle. 
As $|G \Delta G^{**}| =
2\ell$, we can apply the induction hypothesis. We obtain a sequence of
realizations $G=G_0, G_1, \dots, G_k = G^{**}$ with   
$|G_i\Delta G_{i+1}|=4$ and 
$k \leq \frac{1}{2} |G \Delta G^{**}| \leq \frac{1}{2} (|G \Delta
G^{*}|-2)$. Finally, we apply a last swap on the cycle
$(w_1,w_2,w_3,w_4,w_1)$ and thereby transform $G^{**}$ to $G^*$.
In total, the number of swap operations is $k \leq \frac{1}{2} |G \Delta
G^{*}|$.

The case $(w_1,w_4) \not\in A(G) \cap A(G^*)$ is similar. 
This time, we start with a single swap on the cycle $(w_1,w_2,w_3,w_4,w_1)$
and afterwards apply induction to the remaining cycle.
We can treat the case $(w_1,w_2) \in A(G^*)\setminus A(G)$ analogously.
Thus we can exclude the existence of any vertex-disjoint directed alternating
$3$-walk which does not contain at least one arc from
$G\langle V' \rangle$ and its reorientation.
 
It remains to consider the case that there 
is a vertex-disjoint directed alternating
$3$-walk $P = (w_1,w_2,w_3,w_4)$ 
in $G \Delta G^*$ with $(w_1,w_2) \in A(G) \setminus A(G^*)$
or $(w_1,w_2) \in A(G^*) \setminus A(G)$
 but at least one 
arc of $P$ is from $G\langle V' \rangle$ and its reorientation,
say $(v_1,v_2)$.

Recall that $G \Delta G^*$ contains between one and five arcs 
from $G\langle V' \rangle$ and its reorientation.
We distinguish between three cases:

\begin{enumerate}
\item[case I:] $G \Delta G^*$ contains exactly one of these arcs, 
say $(v_1,v_2) \in A(G) \setminus A(G^*)$. (The case that $(v_2,v_1) \in
 A(G^*) \setminus A(G)$ is this special arc can be treated analogously.)
\\
We claim that
the cycle $G \Delta G^*$ must have the form
$(v_1,v_2,v_4,v_5,v_6,v_4,v_5,v_6,v_1)$. Note that $v_4,v_5,v_6$ are
repeated every third step, as otherwise we would obtain a
vertex-disjoint alternating cycle as excluded above.
The cycle cannot be longer than eight, since then we would either obtain a 
vertex-disjoint 3-walk $(v_4,v_5,v_6,v_7)$, also excluded above, or
if $v_4=v_7$ we would violate simplicity of the symmetric difference.
It might be that 
$v_5=v_3$, but $v_4,v_6 \neq v_3$ as otherwise the symmetric
difference would contain more than one arc 
from $G\langle V' \rangle$ and its reorientation.
If $(v_1,v_4) \in A(G) \cap A(G^*)$ there is the alternating directed 
4-cycle $(v_1,v_4,v_5,v_6,v_1)$ which can be swapped. In the remaining
6-cycle the arc $(v_1,v_2)$ is contained, so the induction hypothesis
can be applied.
Otherwise, if $(v_1,v_4) \not\in A(G) \cap A(G^*)$, we first apply the
induction hypothesis to the 6-cycle $(v_1,v_2,v_4,v_5,v_6,v_4,v_1)$,
and afterwards we swap the remaining 4-cycle $(v_1,v_4,v_5,v_6,v_1)$.

\item[case II:] $G \Delta G^*$ contains exactly two of these arcs.\\
Suppose first that these two arcs are adjacent, say $(v_1,v_2),
(v_3,v_2)$. Consider the following arcs $(v_3,v_4), (v_5,v_4),
(v_5,v_6)$ along the symmetric difference. Now $v_5=v_2$ as otherwise
there is an alternating directed walk 
$(v_2,v_3,v_4,v_5)$. Depending whether $(v_5,v_2) \in A(G) \cap A(G^*)$ or not, we can
either swap the alternating 4-cycle $(v_3,v_4,v_5,v_2,v_3)$ or the
remaining part of the symmetric difference together with $(v_5,v_2)$ 
by the induction hypothesis.
Moreover, $v_6=v_3$, as otherwise there would be the vertex-disjoint
alternating directed 3-walk $(v_3,v_4,v_5,v_6)$ excluded above.
But then $(v_3,v_2)$ is also in the symmetric
difference, a contradiction.
Thus, the two arcs from $G\langle V' \rangle$ and its reorientation
are not adjacent. Then, there are at least two other arcs between them
(otherwise the one arc between them would also be from $G\langle V'
\rangle$ and its reorientation). 
By our assumption, there is a vertex-disjoint alternating directed
3-walk $P = (w_1,w_2,w_3,w_4)$ with at least one arc from $G\langle V' \rangle$ and 
its reorientation. In our scenario it must be exactly one such arc. 
Depending whether $(w_1,w_4) \in A(G) \cap A(G^*)$ or not, we can
either swap the alternating 4-cycle $(w_1,w_2,w_3,w_4,w_1)$ or the
remaining part of the symmetric difference together with $(w_1,w_4)$ 
by the induction hypothesis.

\item[case III:] $G \Delta G^*$ contains between three and five 
of these arcs.\\
Suppose first all of them follow consecutively on the alternating
directed cycle. Consider the last two of these arcs, and append the
next arc which must end in a vertex $v_4 \not\in V'$. Then we have
a vertex-disjoint alternating directed 3-walk which contains two arcs
from $G\langle V' \rangle$ and its reorientation, and the remaining
part of the symmetric difference has also such an arc.
Thus we can apply the induction hypothesis and are done.
Otherwise the three to five arcs from $G\langle V' \rangle$ 
and its reorientation are separated. So no alternating directed 
3-walk may contain all of them, in particular not $P$.
We can proceed as in case II).
\end{enumerate}
\vspace*{-2ex}
\end{proof}

\begin{Proposition}\label{prop:arc-swap-6} Let $S$ be a graphical
   arc-swap-sequence and $G$ and $G'$ be two different realizations. 
 If $|G\Delta G'|=6$ and $G\Delta G'$ consists of exactly three
 vertices $V' := \{ v_1,v_2,v_3\}$, 
 then there exist realizations $G_0,G_1,\dots,G_k$ with
 $G_0:=G,G_k:=G'$, $|G_i\Delta G_{i+1}|=4$ and $k \leq 2n+2.$
 \end{Proposition}

\begin{proof}
Since $S$ is an arc-swap-sequence,
Theorem~\ref{th:characterizing-arc-swap-sequences}
implies the existence of a realization $G^*$ such that $G \Delta G^*$
is a simple symmetric cycle and $G^*\langle V' \rangle$ is not a
directed $3$-cycle. By Lemma~\ref{lemma:arc-swap-simple-symmetric},
there are realizations $G_0, G_1, \dots, G_{k'} := G^*$ with 
$|G_i\Delta G_{i+1}|=4$ and $k' \leq \frac{1}{2} |G \Delta G^*| \leq
n$, since $G \Delta G^*$ is simple.
Moreover, we have $|G^* \Delta G'| \leq |G \Delta G^*| + 4 \leq 2n+4$ since
$G$ and $G'$ differ only in their orientation of the $3$-cycle induced
by $V'$. 
The symmetric difference $G^* \Delta G'$ is not necessarily a simple
symmetric cycle,
but can be decomposed into simple symmetric cycles, each containing
at least one arc from $G\langle V' \rangle$ 
and its reorientation.  
On each of these simple symmetric cycles we apply our 
auxiliary Lemma~\ref{lemma:arc-swap-simple-symmetric}. We obtain a
sequence $G^*:= G'_0, \dots, G'_{k''}:= G'$ with 
$|G_i\Delta G_{i+1}|=4$ and $k'' \leq n+2$.
Combining both sequences we obtain a sequence with $k = k' + k'' \leq 2n+2$.
\end{proof}

\begin{Lemma}\label{th:arc-swap-main}
 Let $S$ be a graphical arc-swap-sequence, and $G$ and $G'$ be two
 different realizations. Then there exist realizations $G_0,G_1,\dots,G_k$ with $G_0:=G$, $G_k:=G'$ and $|G_i\Delta G_{i+1}|=4,$ 
where $k\leq \left(\frac{1}{2} |G\Delta G'|-1\right) \cdot (n+1)$. 
\end{Lemma}

\begin{proof} We prove the lemma by induction according to the
  cardinality of the symmetric difference $|G\Delta G'|=2\kappa.$ For
  $\kappa:=2$ we get $|G\Delta G'|=4.$ The correctness of our claim
  follows with $G_1:=G'$. 

For $\kappa:=3$ we distinguish two cases. If $G \Delta G'$ consists of
exactly three vertices, then by Proposition~\ref{prop:arc-swap-6}
we get  a sequence of
  realizations $G_0,G_1,\dots,G_k$ and $k \leq 2n+2 = 2 (n+1)$, as
  claimed.
Otherwise, the symmetric difference $G \Delta G'$ consists of more
than three vertices. By Proposition~\ref{PR1}, case a), there are
 realizations $G_0,G_1,G_2=G'$. 

We assume the correctness of our induction 
  hypothesis for
  all $\kappa\leq \ell.$ Let $|G\Delta G'|=2\ell+2.$ We can assume
  that $\kappa>3.$ 
Suppose first that 
the symmetric difference $G \Delta G'$ decomposes into $t$ simple
symmetric cycles $|(G \Delta G')_i| = 6$. 
Suppose further that all these $(G \Delta G')_i$ consist of exactly
three vertices. 
Clearly, $|G \Delta G'|
= 6t$. We apply our Proposition~\ref{prop:arc-swap-6} to each of these
$t$ cycles one after another and get a sequence of realizations
$G_0, G_1, \dots, G_k=G'$ with $k \leq 2t (n+1) \leq (3t-1)(n+1)= 
\left(\frac{1}{2} |G\Delta G'|-1\right) \cdot (n+1)$.

Otherwise, there is a $(G \Delta G')_1$ which contains at least four
vertices. Swapping the arcs in $(G \Delta G')_1$ leads to a
realization $G^*$.
 By Proposition~\ref{PR1}, there are realizations $G=G_0,
G_1, G_2 = G^*$ with $|G_i \Delta G_{i+1}| = 4$. We can apply the
induction hypothesis on the remaining part of the symmetric difference.
Obviously, we obtain the desired bound in this case.

It remains the case that there exists a simple symmetric 
cycle $(G \Delta G')_1$
of $G \Delta G'$ with $|(G \Delta G')_1| \neq 6$. 
If $|(G \Delta G')_1| = 4$, we use a single swap on $(G \Delta G')_1$
and obtain a realization $G^*$, where $|G^* \Delta G'| = |G \Delta G'|-4$.
By the induction hypothesis, there is a sequence of realizations
$G^*= G_1, G_2 \dots, G_k=G'$ with $k-1 \leq \left(\frac{1}{2}
  |G^*\Delta G'|-1\right) \cdot (n+1) =  \left(\frac{1}{2}
  |G\Delta G'|-3\right) \cdot (n+1)$.
Otherwise, $|(G \Delta G')_1| \geq 8$. 
Using Proposition~\ref{PR2}, we may assume 
that there exists a vertex-disjoint directed 
alternating walk $P = (v_1,v_2,v_3,v_4)$ in $(G \Delta G')_1$
 with $(v_1,v_2),(v_3,v_4)\in A(G)$ and
$(v_3,v_2)\in A(G')$, for the same reasons as in the 
proof of Lemma~\ref{TH1}. 
 
\begin{enumerate}
 \item[case 1:] Assume $(v_1,v_4)\in A(G')\setminus A(G).$\\
This implies $(v_1,v_4)\in G\Delta G'$.
$G_1:=(G\setminus \{(v_1,v_2),(v_3,v_4)\})\cup \{(v_3,v_2),(v_1,v_4)\}$ is a realization of $S$ and it follows
$|G\Delta G_1|=4$ and $|G_1\Delta G'|=2\ell+2-4=2(\ell-1).$ 
 Therefore, we can apply the induction hypothesis on $|G_1\Delta G'|$. 
Thus, we obtain realizations $G_1,G_2,\dots,G_k$ with $G_k:=G'$ and
$|G_i \Delta G_{i+1}|=4$.
where $k-1\leq \left(\frac{1}{2}|G_1\Delta G'|-1\right) \cdot (n+1)
= \left(\frac{1}{2}|G\Delta G'|-3\right) \cdot (n+1)$.

\item[case 2:] Assume $(v_1,v_4)\in A(G)\cap A(G').$\\
This implies $(v_1,v_4)\notin G\Delta G'.$
We construct a new
alternating cycle $C^*:=((G \Delta G')_1\setminus P)\cup \{(v_1,v_4)\}$
with length $|C^*|=|(G \Delta G')_1|-2.$ 
We swap the arcs in $C^*$ and get a realization $G^*$ of S with
$|G_0\Delta G^*|=|C^*|\leq 2\ell$ and $|G^*\Delta G'|=|G \Delta
G'|-(|C^*|-1)+1\leq 2\ell.$ 
According to the induction hypothesis there exist sequences
$G_0^{1}:=G,G_1^{1},\dots,G_{k_1}^1:=G^*$ and
$G_0^{2}:=G^*,G_1^{2},\dots,G_{k_2}^2=G'$ with $k_1\leq
\left(\frac{1}{2}|G_0\Delta G^*|-1\right) \cdot (n+1)$ and 
$k_2\leq \left(\frac{1}{2} | G^* \Delta
G'| -1\right) \cdot(n+1)$. 
We arrange these sequences one after another and get a
sequence with $k=k_1+k_2=\left(\frac{1}{2}|G_0
\Delta G^*|-1 + \frac{1}{2} | G^* \Delta G'| -1\right) \cdot (n+1)
=\left(\frac{1}{2}|G \Delta G'|-1\right) \cdot (n+1)$.

\item[case 3:] Assume $(v_1,v_4)\in A(G)\setminus A(G').$\\
This implies $(v_1,v_4)\in G\Delta G'.$
Note that the arc $(v_1,v_4)$
does not belong to $(G \Delta G')_1$. Therefore, there exists 
an alternating sub-cycle $C^*:= (G \Delta G')_1 \cup \{ (v_1,v_4)\} \setminus P$ 
formed by arcs in $G\Delta G'$. 
We swap the arcs in $C^*$ and get a realization $G^*$ of S with 
$|G_0\Delta G^*|=|C^*|\leq 2\ell$ and $|G^*\Delta G'|=|G \Delta
G'|-|C^*|\leq 2\ell$. According to the induction hypothesis there
exist sequences $G_0^{1}:=G,G_1^{1},\dots,G_{k_1}^1:=G^*$ and
$G_0^{2}:=G^*,G_1^{2},\dots,G_{k_2}^2=G'$ with $k_1\leq
\left(\frac{1}{2}|G_0\Delta G^*|-1\right)\cdot (n+1)$ and 
$k_2\leq \left(\frac{1}{2} (|G \Delta G'|
- |C^*|)-1\right)\cdot (n+1) $. We arrange these sequences
one after another and get a sequence with 
 $k=k_1+k_2=\left(\frac{1}{2}|G_0\Delta G^*|-1+ \frac{1}{2} (|G \Delta G'|
- |C^*|)-1\right)\cdot (n+1) = \left(\frac{1}{2}(|G\Delta G'|)
-2\right) \cdot (n+1)$.

\item[case 4:] Assume $(v_1,v_4)\notin A(G)\cup A(G')$.
This implies $(v_1,v_4)\notin G\Delta G'.$
It exists the alternating cycle $C:=(P,(v_1,v_4))$ with $(v_1,v_4)\notin A(G).$ $G_1:=(G_0\setminus \{(v_1,v_2),(v_3,v_4)\})\cup \{(v_3,v_2),(v_1,v_4)\}$ is a realization of $S$ and it follows
$|G_0\Delta G_1|=4$ and $|G_1\Delta G'|=2\ell+2-2=2\ell.$ 
According to the induction hypothesis there exist realizations
$G_1,G_2,\dots,G_k$ with $G_k:=G'$  where $k_1:=1$ and 
$k_2:=k-1\leq \left(\frac{1}{2}|G_1\Delta G'|-1 \right) \cdot (n+1).$
Hence, we get the sequence $G_0,G_1,\dots,G_k$ with $k=k_1+k_2= 1+
\left(\frac{1}{2}|G_1\Delta G'|-1\right) \cdot (n+1)
\leq \left(\frac{1}{2}|G\Delta G'|-1\right) \cdot (n+1)$.
\end{enumerate}
\vspace*{-2ex}
\end{proof}

\begin{Corollary}\label{cor:R_ss-strongly-connected}
 State graph $\overline{\Phi}$ is a strongly connected directed graph if and only if a given sequence $S$ is an arc-swap-sequence.
\end{Corollary}

An arc-swap-sequence implies the connectedness of the simple
realization graph $\overline{\Phi}$. 
Therefore, for such sequences we are able to
make random walks on the simple state graph $\overline{\Phi}$ which can
be implemented easily. 
We simplify the random walk Algorithm~\ref{alg:directed} for
arc-swap-sequences in using realization graph $\overline{\Phi}.$ Hence, our
data structure $DS$ only contains pairs of non-adjacent arcs. We can
ignore lines $12$ to $20$ in Algorithm~\ref{alg:directed}. We
denote this modified algorithm as the
\emph{Arc-Swap-Realization-Sample} Algorithm $3$.

\begin{Theorem} Algorithm $3$ is a random walk on the state graph
  $\overline{\Phi}$ which uniformly samples a directed graph $G'=(V,A)$
  as a realization of an arc-swap-sequence $S$  for $\tau\rightarrow \infty.$ 
 \end{Theorem}

\begin{proof}
 Algorithm $3$ chooses all elements in $DS$ with the same constant
 probability. For a vertex $V_G\in V_{\overline{\Phi}}$ there exist for all these
 pairs of arcs in $A(G')$ either incoming and outgoing arcs on
 $V_{G'} \in V_{\overline{\Phi}}$ or a loop. 
We get a transition matrix $M$ for $\overline{\Phi}$ with
$p_{ij}=\frac{1}{d}$ for $i,j\in A_{\overline{\Phi}},i\neq j$,
$p_{ij}=1-\sum_{\{i|(i,j)\in A_{\overline{\Phi}}\}}\frac{1}{d}$ for
$i,j\in V_{\phi},i=j,$ otherwise we set $p_{ij}=0$ where $d:={|A(G)|\choose
  2}-2\sum_{i=1}^{n}{a_i \choose 2}-\sum_{i=1}^{n}a_i b_i.$ 
Since, $\overline{\Phi}$ is a regular, strongly connected, symmetric, and non-bipartite directed graph, the distribution of all realizations in a $t$th step converges asymptotically to the uniform distribution, see Lovasz~\cite{Lovasz96}.
\end{proof}

\subsection{Practical Insights And Applications}\label{Sec:PracticalInsights} 

As mentioned in the Introduction, many ``practitioners'' use the
switching algorithm for the purpose of network analysis, regardless
whether the corresponding degree sequence is an arc-swap-sequence or not.
In this section we would like to discuss under which circumstances
this common practice can be well justified and when it may lead to
wrong conclusions.

What would happen if we sample
using the state graph $\overline{\Phi}$ for a sequence $S$
which is not an arc-swap-sequence? 
Clearly, we get the insight that $\overline{\Phi}$ has several
connected components, but as we will see $\overline{\Phi}$ consists of
at most $2^{\lfloor\frac{|V|}{3}\rfloor}$ isomorphic components containing
exactly the same realizations up to
the orientation of directed $3$-cycles each consisting of an induced cycle set $V'$. Fortunately, we can identify all induced cycle sets using our results in Theorem~\ref{th:characterizing-arc-swap-sequences} by only considering an arbitrary realization $G$. 

\begin{Proposition}\label{Pr2}
 Let $S$ be a graphical sequence which is not an arc-swap-sequence and
 has at least two different induced cycle sets $V'$ and $V''.$ Then it follows $V'\cap V''=\emptyset.$
\end{Proposition}

\begin{proof}
Without loss of generality we can label the vertices in $V'$ with $v'_1,v'_2,v'_3$ and in $V''$ with $v''_1,v''_2,v''_3.$ Let $G$ be a realization where $\{(v_1',v'_2),(v_2',v_3'),(v_3',v_1'),(v_1'',v_2''),(v_2'',v_3''),(v_3'',v_1'')\}\subset A(G).$ We distinguish between two cases.
\begin{enumerate}
\item[a):] Assume $|V'\cap V''|=1$ where $v'_1=v''_1.$
 If it exists arc $(v_3'',v_3')\in A(G)$ we find the alternating $4$-cycle $(v_3'',v_3',v_1',v_2'',v_3'')$ which implies a new realization $G^*$ where $G^*\left\langle V' \right\rangle$ is not a directed cycle in contradiction to our assumption that $V'$ is an induced cycle set. Hence, it follows $(v_3'',v_3')\notin A(G).$ In this case we find the alternating cycle $(v_3'',v_1',v_2',v_3',v_3'').$
\item[b):] Assume $|V'\cap V''|=|\{v'_1,v'_2\}|=2$ where $v'_1=v_1''$
  and $v'_2=v_2''.$ If  arc $(v_3'',v_3')\notin A(G)$ exists we find the alternating $4$-cycle $(v_1',v_3'',v_3',v_2',v_1')$ which implies a new realization $G^*$ where $G^*\left\langle V' \right\rangle$ is not a directed cycle in contradiction to our assumption that $V'$ is an induced cycle set. Hence, it follows $(v_3'',v_3')\in A(G).$ In this case we find the alternating cycle $(v_3'',v_3',v_1',v_2',v_3'').$
\end{enumerate}
\vspace*{-2ex}
\end{proof}

As the induced $3$-cycles which appear in every realization are
vertex-disjoint, we can reduce the in- and out-degrees of all vertices
in these cycles by one, and obtain a new degree sequence which must be
an arc-swap sequence.

\begin{Theorem}
Let $S$ be a sequence. Then the state graph $\overline{\Phi}$ consists of at most $2^{\lfloor\frac{|V|}{3}\rfloor}$ isomorphic components.
 \end{Theorem}

\begin{proof}
We assume that $S$ is not an arc-swap-sequence, otherwise we apply
Theorem~\ref{cor:R_ss-strongly-connected} and get a strongly
connected digraph $\overline{\Phi}$. With Proposition \ref{Pr2} it follows the
existence of at most  $\lfloor\frac{|V|}{3}\rfloor$ induced cycle sets
for $S.$ Consider all realizations $G^j$ possessing a fixed
orientation of these induced $3$-cycles which implies $G^j\left\langle
  V_i\right\rangle=G^{j'}\left\langle V_i\right\rangle$ for all such
realizations. We pick out one of these orientation scenarios and
consider the symmetric difference $G^j\Delta G^{j'}$ of two such
realizations. Since, all induced $3$-cycles 
are identical in $G^j$ and $G^{j'}$, we can 
delete each arc 
of these induced cycle sets $V'$  and get
the reduced graphs $G^j_{c}$ and $G^{j'}_{c}.$ Both are realizations
of an arc-swap-sequence $S'$. Applying  Theorem~\ref{th:arc-swap-main} 
we obtain,
that there exist realizations $G_0:=G^j_{c},\dots,G_k:=G^{j'}_{c}$
$|G_i\Delta G_{i+1}|=4$ and $k\leq |G^j_{c}\Delta
G^{j'}_{c}|.$ Hence, each induced subdigraph $\overline{\Phi}\left\langle
  \{V_{G^{j}}|V_{G^{j}}\in V_{\overline{\phi}}\textnormal{ and } G^j\textnormal{ is
    a realization for one fixed orientation scenario}\}\right\rangle$
is strongly connected. On the other hand, we get for each fixed
orientation scenario exactly the same  realizations $G^{j}.$
Since, all induced $3$-cycles are isomorphic, it follows that all
realizations which are only different in the orientation of such
directed $3$-cycles are isomorphic. By
Theorem~\ref{th:characterizing-arc-swap-sequences}, 
there does not exist an alternating cycle destroying an induced
$3$-cycle. Hence, the state graph $\overline{\Phi}$ consists of
exactly $2^k$ strongly connected isomorphic components where $k$ 
is the number of induced cycle sets $V'.$
\end{proof}

\paragraph{Applications in Network Analysis}

Since the switching algorithm samples only in one single component of 
$\overline{\Phi}$, one has to be careful to get the correct
estimations for certain network statistics.
For network statistics on unlabeled graphs, it suffices 
to sample in a single component which
reduces the size of $V_{\overline{\Phi}}$ by a factor $2^k$, 
the number of components in
$\overline{\Phi}$,
where $k$ is the number of induced cycle sets of the prescribed degree
sequence.
Examples where this approach is feasible are 
network statistics like
the average diameter or the motif content over all
realizations.

For labelled graphs, however, the random walk on $V_{\overline{\Phi}}$
systematically over- and under-samples the probability that an arc is
present. 
Suppose that the random walk starts  with a realization $G=(V,A)$.
If an arc $(v_1,v_2) \in A(G)$ belongs to an induced cycle set,
it appears with probability 1 in all realizations of the random walk.
The opposite arc $(v_2,v_1) \not\in A(G)$, will never occur.
In an unbiased sampling over all realizations, each of these arcs,
however, occurs with probability $1/2$. 
All other arcs occur with the same probability in a single component
of $V_{\overline{\Phi}}$ as in the whole state graph.   
This observation can be used to compute correct probabilities for all arcs.

\section{Concluding Remarks}

In this paper, we have presented Markov chains
for sampling uniformly at random undirected and directed graphs with a
prescribed degree sequence. 
The key open problem remains to analyze whether these Markov chains
are rapidly mixing or not.

\newcommand{\etalchar}[1]{$^{#1}$}
\providecommand{\bysame}{\leavevmode\hbox to3em{\hrulefill}\thinspace}
\providecommand{\MR}{\relax\ifhmode\unskip\space\fi MR }
\providecommand{\MRhref}[2]{%
  \href{http://www.ams.org/mathscinet-getitem?mr=#1}{#2}
}
\providecommand{\href}[2]{#2}

\addtolength{\baselineskip}{0.3mm}

\newpage

\appendix

\centerline{\Large \bf Appendix}

\section{Further Examples Where Switching Fails}

As we have seen in Example~\ref{example:disconnectedness}, the
switching algorithm will fail in general.
Here we give further non-trivial classes of graphs where it also fails.
All problematic
instances are realizations which are different in at least one
directed $3$-cycle but not all of them are not changeable with
alternating $4$-cycles. 

Consider the following Figures \ref{PEineRichtung} and \ref{PDreiPartition}. Both examples cannot be changed to a realization which is only different in the orientation of the directed $3$-cycle by a sequence of alternating $4$-cycles.

\begin{figure}[hbp]
\centering
\includegraphics[width=7cm]{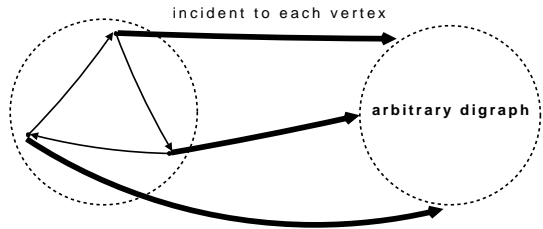}
\caption{\label{PEineRichtung} All vertices in a $3$-cycle are incident in one direction with vertices in an arbitrary subdigraph.}
\end{figure}

\begin{figure}[hbp]
\centering
\includegraphics[width=7cm]{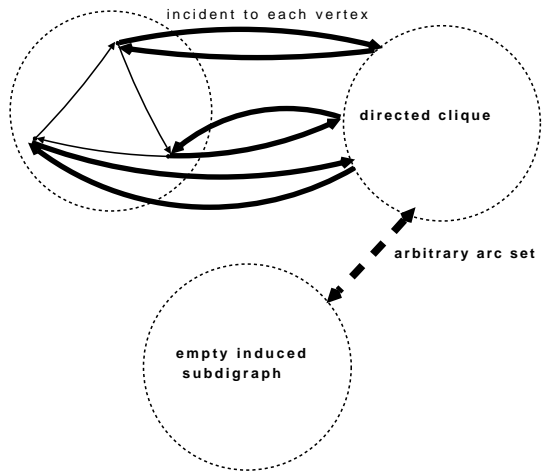}
\caption{\label{PDreiPartition} All vertices in a $3$-cycle are incident in both directions with a directed clique. An independent set of vertices is arbitrarily incident with the directed clique.}
\end{figure}

\end{document}